\documentclass[11pt]{amsart}
\usepackage{amsmath,stmaryrd,amsthm,tikz,calc}
\usepackage{color}
\usepackage[margin=1in]{geometry}
\usepackage{color}
\usepackage{moreverb}
\usepackage{graphicx} 
\usepackage{float}
\usepackage{pdfsync}
\usepackage{hyperref}  
\hypersetup{colorlinks=true}    
\usepackage{bbm}
\usepackage{tcolorbox}

\date{\today}

\newcommand{\Rl}{\mathbb{R}}

\newcommand{\N}{\mathbb{N}}

\newcommand{\A}{\mathcal{A}}

\newcommand{\cP}{\mathcal{P}}

\newcommand{\be}{\begin{equation}}
\newcommand{\ee}{\end{equation}}
\newcommand{\bea}{\begin{eqnarray}}
\newcommand{\eea}{\end{eqnarray}}
\newcommand{\beann}{\begin{eqnarray*}}
\newcommand{\eeann}{\end{eqnarray*}}
\newcommand{\eq}[1]{(\ref{#1})}

\newtheorem{theorem}{Theorem}[section]
\newtheorem{proposition}[theorem]{Proposition}
\newtheorem{corollary}[theorem]{Corollary}
\newtheorem{lemma}[theorem]{Lemma}
\newtheorem{assumption}[theorem]{Assumption}

\newtheorem{remark}[theorem]{Remark}

 \numberwithin{equation}{section}

\renewcommand{\epsilon}{\varepsilon}

 \def\idty{{\mathchoice {\mathrm{1\mskip-4mu l}} {\mathrm{1\mskip-4mu l}} %
{\mathrm{1\mskip-4.5mu l}} {\mathrm{1\mskip-5mu l}}}}

\newcommand\supp{\text{supp}}

\newcommand{\beq}{\begin{equation}}
\newcommand{\eeq}{\end{equation}}

\newcommand{\R}{\mathbb{R}}

\let\<\langle
\let\>\rangle

\renewcommand{\d}{\mathrm{d}}  
\providecommand{\wtilde}[1]{\widetilde{#1}}

\usepackage{todonotes}

\newenvironment{remarks}{\begin{myremarks}\begin{nummer}}%
    {\end{nummer}\end{myremarks}}
    {\end{nummer}\end{myexamples}}

\newcounter{numcount}
\newcommand{\labelnummer}{(\roman{numcount})}%

\makeatletter
\providecommand{\showkeyslabelformat}[1]{\relax}        
\let\mysaveformat\showkeyslabelformat                   %
\def\myformat#1{\raisebox{-1.5ex}{\mysaveformat{#1}}}   %

\newenvironment{nummer}%
  {\let\curlabelspeicher\@currentlabel%
    \begin{list}{\textup{\labelnummer}}%
      {\usecounter{numcount}\leftmargin0pt%
        \topsep0.5ex\partopsep2ex\parsep0pt\itemsep0ex\@plus1\p@%
        \labelwidth2.5em\itemindent3.5em\labelsep1em%
      }%
    \let\saveitem\item%
    \def\item{\saveitem%
      \def\@currentlabel{\curlabelspeicher\kern.1em\labelnummer}}%
    \let\savelabel\label%
    \def\label##1{{\ifnum\thenumcount=1\let\showkeyslabelformat\myformat\fi\savelabel{##1}}%
										{\def\@currentlabel{\labelnummer}%
									 	\let\showkeyslabelformat\@gobble
									 	\savelabel{##1item}%
										}%
	   							}%
  }{\end{list}}%

  {\let\curlabelspeicher\@currentlabel%
    \begin{list}{\textup{\labelnummer}}%
      {\usecounter{numcount}\leftmargin0pt%
        \topsep0.5ex\partopsep2ex\parsep0pt\itemsep0ex\@plus1\p@%
        \labelwidth2.5em\itemindent0em\labelsep1em%
        \leftmargin2.5em}%
    \let\saveitem\item%
    \def\item{\saveitem%
      \def\@currentlabel{\curlabelspeicher\kern.1em\labelnummer}}%
    \let\savelabel\label%
    \def\label##1{{\ifnum\thenumcount=1\let\showkeyslabelformat\myformat\fi\savelabel{##1}}%
										{\def\@currentlabel{\labelnummer}%
									 	\let\showkeyslabelformat\@gobble
									 	\savelabel{##1item}%
										}%
    							}%
  }{\end{list}}%

\usepackage{enumitem}

\let\OldItem\item
\newcommand{\MyItem}[2][]{}%
\newtheorem{myremarks}[theorem]{Remarks}

\begin{document}

\title[Many-body fermion dynamics]{Lieb-Robinson bounds and strongly continuous dynamics\\ for a class of many-body fermion systems in $\Rl^d$.}


\author[M. Gebert]{Martin Gebert}
\address{Department of Mathematics\\
University of California, Davis\\
Davis, CA 95616, USA}
\email{mgebert@math.ucdavis.edu}
\author[B. Nachtergaele]{Bruno Nachtergaele}
\thanks{Based upon work supported by the National Science Foundation under Grant DMS-1813149.}
\address{Department of Mathematics and Center for Quantum Mathematics and Physics\\
University of California, Davis\\
Davis, CA 95616, USA}
\email{bxn@math.ucdavis.edu}
\author[J. Reschke]{Jake Reschke}
\address{Department of Mathematics\\
University of California, Davis\\
Davis, CA 95616, USA}
\email{jreschke@math.ucdavis.edu}
\author[R. Sims]{Robert Sims}
\address{Department of Mathematics \\
University of Arizona\\
Tuscon, AZ 85721, USA}
\email{rsims@math.arizona.edu}

\begin{abstract}
We introduce a class of UV-regularized two-body interactions for fermions in $\Rl^d$ and prove a Lieb-Robinson estimate
for the dynamics of this class of many-body systems. As a step toward this result, we also prove a propagation bound of Lieb-Robinson type 
for Schr\"odinger operators. We apply the propagation bound to prove the existence of infinite-volume dynamics as a strongly continuous 
group of automorphisms on the CAR algebra. 
\end{abstract}
\maketitle

\section{Introduction}\label{sec:introduction}
The goal of this paper is to study propagation estimates for interacting fermion systems in $\Rl^d$, $d\geq 1$, and to apply them to
construct the infinite-volume dynamics for a class of such systems as a strongly continuous one-parameter group of 
automorphisms of the standard CAR algebra. We introduce a class of short-range, UV-regularized two-body interactions for which
this is possible. Without the use of a UV cut-off of some kind, such a result cannot be expected to hold. See, for example, the discussion in 
\cite[Introduction to Section 6.3]{bratteli:1997}. Nevertheless, as Sakai notes in the last paragraph of his book \cite{sakai:1991}, constructing 
the dynamics for interacting systems is one of the most important problems. To address this problem, a common approach is to consider
the dynamics in representations of the algebra of observables associated with a class of sufficiently regular states. This is not our approach here.
Instead we introduce a UV regularization of the interactions. This allow us to construct the infinite system dynamics as 
automorphisms of the CAR algebra of observables that depend continuously on time.
A typical situation where it is advantageous to consider the dynamics on the observables algebra of the infinite system is in non-equilibrium statistical mechanics, 
where until now one would either use a quasi-free dynamics (as, e.g., in\cite{frohlich:2003}) or 
work in a lattice setting where UV regularization is provided by the lattice  (see, e.g., \cite{hastings:2015,grundling:2017,bachmann:2017b,bachmann:2019}, 
and \cite{robinson:1968,ruelle:1969} for the original and fundamental existence result for quantum spin systems.).

One broad class of models in which UV degrees of freedom are naturally absent are mean field models and 
related limiting regimes and the dynamics of such models have been studied including in infinite volume. For example, well-posedness for the Hartree equation in infinite volume, which describes the mean field limit \cite{elgart:2004}, 
has been proved by Lewin and Sabin in \cite{Lewin:2015}.

The regularization we adopt in this paper is smearing the interactions by Gaussians parameterized by $\sigma >0$ in such a way that the pair interaction between point 
particles is recovered in the limit $\sigma \to 0$ (See Appendix \ref{sec:strong_resolvent_limit} for a proof). Formally, in second quantization, this leads to a Hamiltonian of the form
\be
H^\sigma_\Lambda=\int_{\Rl^d} (\nabla a^*_x \nabla a_x + V(x) a^*_x a_x) dx + \frac{1}{2} \int_\Lambda\int_\Lambda W(x-y) a^*(\varphi_x^\sigma)a^*(\varphi_y^\sigma)a(\varphi_y^\sigma)a(\varphi_x^\sigma) dx dy,
\label{smearedHamiltonian}\ee
where $V$ is an external potential such as a smooth periodic function, and $W$ is a short-range two-body interaction. We defer stating precise conditions on
$V$ and $W$ until Section \ref{sec:results}. The smearing is only needed in the interaction and one can take for $\varphi_x^\sigma$ an $L^1$-normalized Gaussian of width $\sigma$ and  centered at $x\in\Rl^d$. The parameter $\sigma$ can be interpreted as the size of the particles and, as discussed in Appendix \ref{sec:strong_resolvent_limit}, restricted to the 
$N$-particle Hilbert space, for any finite number of particles $N$, in either a finite or infinite volume, the dynamics converges to the standard Schr\"odinger dynamics
generated by the self-adjoint Hamiltonian $H_N$ given by
\be
H_N = \sum_{k=1}^N (-\Delta_k + V(x_k)) +\sum_{1\leq k < l\leq N}  W(x_k-x_l).
\ee

Having a state-independent definition of the dynamics has both conceptual and practical advantages. From early on it was realized however that the subtle, non-robust, 
property of (thermodynamic) stability may be an obstacle to using perturbation series to define Heisenberg dynamics for infinite systems in the continuum \cite{dyson:1952}.
Therefore, it is not surprising that attempts were made to construct toy models of interacting theories for which stability could be proved. An
early example is \cite{streater:1968}.  In \cite{streater:1970} an infinite-volume dynamics for interacting fields is obtained using relativistic locality (Minkowski space). 
The only previous Euclidean construction of infinite-system dynamics on the CAR algebra over $L^2(\R^d)$ that explicitly considers a regularized pair interaction, as far as we are aware, is by Narnhofer and Thirring \cite{narnhofer:1990}. In that work the authors were motivated by the desire to preserve the 
Galilean invariance 
of the dynamics, which led them to employ a somewhat contrived UV regularization.
The smearing of the form \eq{smearedHamiltonian} used here is, we believe, more natural and likely to faithfully reproduce the low-energy physics.

Before summarizing our results, we point out that {\em defining} a dynamics on the CAR algebra over $L^2(\R^d)$ is by itself not the issue. Including pair interactions 
in a densely defined self-adjoint Hamiltonian on Fock space has been accomplished a long time ago. The corresponding one-parameter group of unitaries can be used
to define a dynamics as a group of automorphisms on the bounded operators on Fock space, which includes the CAR algebra. This dynamics, however,  is in general {\em not strongly 
continuous}. This is because the commutator of the unregularized interaction term with a creation or annihilation operator is unbounded.

Our proof of convergence of the thermodynamic limit of the infinite-volume dynamics hinges on a propagation estimate of Lieb-Robinson type \cite{lieb:1972} for systems in which the 
interaction is only active in a bounded volume $\Lambda$, with estimates that are uniform in $\Lambda$. Let $\tau_t^\Lambda(\cdot)$ denote the Heisenberg evolution  with the interactions restricted to $\Lambda$ (see (\ref{fv_pert_dyn}) for the precise definition) and define the one-particle Schr\"odinger evolution in the usual way:
\be
f_t = e^{-it(-\Delta +V)} f,\quad  f\in L^2(\Rl^d),\quad  t\in \Rl.
\ee 
\noindent
\textbf{Lieb-Robinson Bound for Schr\"odinger operators.}
Let $V$ be given as the Fourier transform of a finite Borel measure of compact support on $\Rl^d$. For $\sigma >0$ and $x\in\Rl^d$, denote 
by $\varphi_x^\sigma$ the $L^1$ normalized Gaussian on $\Rl^d$ with mean $x$ and variance $\sigma$. Then, there exist constants $C_1,C_2,C_3>0$, such that for all 
$f\in L^2(\Rl^d)$ and $t\in\Rl$
one has
\begin{equation} \label{ft_to_Gauss_est}
\big|\<e^{-it (-\Delta+V)}f, \varphi^\sigma_x\>\big|
\leq C_1  e^{C_{2} |t| \ln |t|} \int_{\R^d} \d y\,e^{-\frac { C_{3} }{t^2+ 1}  |x-y|} |f(y)|.
\end{equation}

A more detailed estimate and explicit constants are given in Proposition \ref{prop:LRB_SO} and Corollary \ref{cor:LRB_SO_simple}.
For discrete Schr\"odinger operators on graphs a Lieb-Robinson type propagation estimate holds for any real-valued diagonal potential \cite{aizenman:2012}.

Let $\tau_t^\Lambda$, be the Heisenberg dynamics generated by $H_\Lambda$ in \eq{smearedHamiltonian} for bounded $\Lambda\subset\Rl^d$, and $t\in\Rl$. We will prove the following result as Theorem \ref{thm:propagationbound}.

\noindent
\textbf{Propagation Bound for many-body fermion dynamics.}
{\em Let $W\in L^\infty(\Rl^d)$ be real-valued and satisfying $W(-x)=W(x)$ and $|W(x)| \leq Ce^{-a |x|}$, 
for some $C, a >0$.
Then, there exist continuous functions $C(t), a(t) >0$ such that for all bounded and measurable $\Lambda\subset\Rl^d$, 
and $f,g\in L^1(\Rl^d) \cap L^2(\Rl^d)$, one has the following bounds:
\begin{align}
\Vert \{ \tau^\Lambda_t(a(f)), a^*(g)\} - \langle  e^{-it (-\Delta+V)}f,g\rangle\idty  \Vert
&\leq \Vert f\Vert_1 \Vert g\Vert_1 e^{C(t)} e^{-a(t) d(\supp(f),\supp(g))}\label{basicbound}\\
\Vert \{ \tau^\Lambda_t(a(f)), a(g)\}\Vert
&\leq \Vert f\Vert_1 \Vert g\Vert_1 e^{C(t)} e^{-a(t) d(\supp(f),\supp(g))}\label{basicbound_b}
\end{align}
where $d(\supp(f),\supp(g))$ denotes the distance between the essential supports of $f$ and $g$.
}

Explicit forms of $C(t)$ and $a(t)$ are given in Theorem \ref{thm:propagationbound}. 
This Lieb-Robinson type bound provides localization estimates for general elements in the CAR algebra by the usual algebraic relations in the same way as for lattice fermion systems as in
\cite{hastings:2006,nachtergaele:2016b,bru:2017}.

As an application of this propagation bound above, which is of independent interest, we then prove the existence and continuity of the infinite systems dynamics. 
See Theorem \ref{thm:main} for the precise statement. There are other approaches to proving the convergence of the dynamics
in the thermodynamic limit. Using propagation bounds, however, yields a short and intuitive proof.

\noindent
\textbf{Strongly continuous infinite-volume dynamics.}
{\em There exists a strongly continuous one-parameter group of automorphism of the CAR algebra over $L^2(\Rl^d)$, $\{\tau_t\}_{t\in\Rl}$,
such that }
\be
\lim_{\Lambda\uparrow \Rl^d} \tau^\Lambda_t(a(f)) = \tau_t(a(f)), \text{ for all }f\in L^2(\Rl^d).
\ee

The strategy for proving existence of the thermodynamic limit of the Heisenberg dynamics using propagation bounds appears to be quite general and has been employed 
successfully for lattice systems \cite{bratteli:1997,nachtergaele:2010,nachtergaele:2019}. This method works whenever the interactions restricted to a bounded region are 
described by a bounded self-adjoint operator. It is worth noting that the free part of the dynamics does not require a cut-off for this result to hold. Due to its uniformity in $\Lambda$, 
the propagation bound \eq{basicbound} extends to the infinite system dynamics.

Several generalizations of the propagation bounds could be considered. For Schr\"odinger operators, we expect that the restrictions on $V$ can be relaxed. 
The many-body bounds are derived here for regularized pair interactions only. Our approach can handle $k$-body terms with virtually no changes. 
A different type of extension of obvious interest would be to consider fermions in an external magnetic field. In contrast, constructing the many-body dynamics 
for boson systems one has to face an additional element of unboundedness that has long been understood to force one to consider a weaker topology to express the continuity in time
 \cite{verbeure:2011}. Already for boson lattice systems, such as oscillator lattices, Lieb-Robinson bounds can be derived but one finds bounds that are no longer in terms of the operator norm 
of the observables \cite{amour:2010,nachtergaele:2009a}. Such bounds can nevertheless still be used to prove the existence of infinite-systems dynamics \cite{nachtergaele:2010}. Another approach to define the dynamics of infinite oscillator lattices was developed by Buchholz \cite{buchholz:2017}, who constructs a strongly continuous dynamics on 
the Resolvent Algebra \cite{buchholz:2008}.

The existence of propagation bounds of Lieb-Robinson type and the strongly continuous infinite-volume dynamics for many-body systems with Hamiltonians of the
form \eq{smearedHamiltonian} provide a new avenue for applications. For example, if we choose for $V$ a periodic potential, such that $-\Delta +V$ has a band 
structure with a gap, the non-interacting many-body ground state at suitable fermion density is gapped. We expect this gap to persist in the presence of interactions
as in \eq{smearedHamiltonian} with $W$ sufficiently small. Stability of the ground state gap has been proved for broad classes of lattice systems \cite{bravyi:2010,klich:2010,bravyi:2011,schuch:2011a,michalakis:2013,Cirac:2013,de-roeck:2019,frohlich:2018, hastings:2019,nachtergaele:inprep}. We believe that 
an analogous result for the continuum systems studied in this paper is now within reach.

\section{Model and statement of main results}\label{sec:results}

Let $d\geq 1$ and take $\Delta$ to be the Laplace operator on $\Rl^d$. For any real-valued $V\in L^\infty(\R^d)$, we will
denote by 
\beq\label{def:H_1}
H_1=-\Delta+V
\eeq
the corresponding (self-adjoint) Schr\"odinger operator with domain $H^2(\Rl^d)\subset L^2(\R^d)$, see \cite{reed:1975} for more details. 
As required, we will impose further conditions on $V$, e.g. see \eqref{def:V}.

Our goal is to analyze a class of operators on the fermionic Fock space. We will follow closely the notation in 
 \cite{bratteli:1997}, see specifically Section 5.2.1, and refer the reader there for more details. 
Let us denote by
\begin{equation}
\mathfrak{F}^-=\bigoplus_{n=0}^\infty \big( L^2(\Rl^d)^{\otimes n}\big)^-
\end{equation} 
the anti-symmetric Fock space (Hilbert space) generated by $L^2( \mathbb{R}^d)$. In the above,
$L^2(\Rl^d)^{\otimes n}$ is short for $\bigotimes_{k=1}^n L^2(\Rl^d)$ and $(\,\cdot\,) ^- $ denotes anti-symmetrization. 
For each $f\in L^2(\R^d)$, take $a(f) \in \mathcal{B}(\mathfrak F^-)$, the bounded linear 
operators over $\mathfrak F^-$, to be the annihilation operator corresponding to 
$f$, and denote by $a^*(f)$, its adjoint, the corresponding creation operator. 
It is well-known that these creation and annihilation operators satisfy the 
canonical anti-commutation relations (CAR)
\beq\label{CAR}
\big\{ a(f) , a(g)\big\} =0\quad\text{and} \quad \big\{ a(f), a^*(g)\big\} = \<f,g\>\idty \quad \mbox{for all } f,g\in L^2(\R^d) \, 
\eeq
where $\{A,B\} =AB+BA$ denotes the anti-commutator, $\<\cdot,\cdot\>$ the scalar product in $L^2(\R^d)$, and $\idty$ is
the identity acting on $\mathfrak F^-$. In addition, one has that
\begin{equation} \label{a+c_op_norm}
\|a^*(f)\| = \|a(f)\| = \|f\|_2 \quad \mbox{for all } f \in L^2( \mathbb{R}^d) 
\end{equation}
where here, and in the following, $\|\cdot\|_p$ will refer to the $L^p$-norm for $p\in[1,\infty]$ and $\|\cdot\|$ will denote the operator norm.

The models we will consider are defined in terms of a particular class of annihilation and creation operators.
Let $\sigma>0$, take $x \in \mathbb{R}^d$, and consider the Gaussian $\varphi_x^{\sigma} : \mathbb{R}^d \to \mathbb{R}$ with
\begin{equation}\label{Def:Gaussian}
\varphi_x^{\sigma}(y) = \frac{1}{(2 \pi \sigma^2)^{d/2}} e^{- \frac{|y-x|^2}{2 \sigma^2}} \quad \mbox{for all } y \in \mathbb{R}^d \, .
\end{equation}
We say that $\varphi_x^{\sigma}$ is centered at $x\in\R^d$ with variance $\sigma^2$. We have chosen an $L^1$-normalization, i.e.
 $\|\varphi_x^{\sigma} \|_1 = 1$ for all $x \in \R^d$. Given (\ref{a+c_op_norm}), it is clear that for any $x \in \mathbb{R}^d$,
 \begin{equation} \label{Def:C_sig}
 \| a^*( \varphi_x^{\sigma}) \|^2 =  \| a( \varphi_x^{\sigma}) \|^2 =  \| \varphi_x^{\sigma} \|_2^2 =  (4 \pi \sigma^2)^{-d/2} =: C_\sigma
 \end{equation}
where we have introduced the notation $C_\sigma>0$ as this quantity will enter our estimates frequently.

For any bounded and measurable set $\Lambda\subset\R^d$, we will analyze the operator
\begin{equation} \label{Def:H_L_sig}
H_\Lambda^{\sigma}=\d\Gamma(H_1)+W_\Lambda^{\sigma}
\end{equation}
acting on $\mathfrak{F}^-$, where $\d \Gamma (H_1)$ denotes the second quantization of $H_1$, again see \cite{bratteli:1997} for the definition,
and the interaction $W_\Lambda^{\sigma}$ is given by
\begin{equation}\label{def:interaction}
W_{\Lambda}^{\sigma}=\frac{1}{2} \int_{\Rl^{d}} \int_{\Rl^{d}}\d x\,\d y\, W_\Lambda(x,y)a^*(\varphi_x^\sigma)a^*(\varphi_y^\sigma)a(\varphi_y^\sigma)a(\varphi_x^\sigma)
\end{equation}
where $W_\Lambda: \mathbb{R}^d \times \mathbb{R}^d \to \mathbb{R}$ has the form 
$W_{\Lambda}(x,y)=\chi_{\Lambda\times\Lambda}(x,y)W(x-y)$ for some real-valued $W\in L^\infty(\R^d)$ and $\chi_{\Lambda\times\Lambda}(x,y)$ denotes the indicator function of the set $\Lambda\times\Lambda$. In this case, for each fixed $\sigma>0$ and any $\Lambda \subset \mathbb{R}^d$ that is bounded and measurable, 
we have that  
\beq \label{W_L:bd}
\|W_{\Lambda}^{\sigma}\| \leq\frac 1 2 C_{\sigma}^2 \int_{\Rl^{d}} \int_{\Rl^{d}}\d x\,\d y\, |W_\Lambda(x,y)|
\leq\frac 1 2\left( \frac 1 {4\pi \sigma^2} \right)^d \|W\|_\infty |\Lambda|^2
\eeq
where $|\Lambda|$ denotes the Lebesgue measure of $\Lambda$. Thus for bounded $W$, 
$W_\Lambda^\sigma \in \mathcal{B}( \mathfrak{F}^-)$ for each choice of $\sigma>0$ and $\Lambda \subset \mathbb{R}^d$ as above. 
We conclude that, in these cases, $H_\Lambda^{\sigma}$ is a well-defined self-adjoint operator on the anti-symmetric Fock space $\mathfrak{F}^-$.

As we progress, more assumptions will need to be made on $W$. For ease of later reference, we state them here.

\begin{assumption}[On $W$]  \label{Assump:W} Let $W : \mathbb{R}^d \to \mathbb{R}$ satisfy
\begin{enumerate}[label=(\roman*)]
\item $W \in L^{\infty}( \mathbb{R}^d)$ is real-valued;
\item $W$ is symmetric, i.e. $W(-x) = W(x)$ for almost every $x \in \mathbb{R}^d$;
\item $W$ is short-range, i.e. there are positive numbers $a$ and $c_W$ for which
\begin{equation} \label{W:dec}
|W(x)| \leq c_We^{-a|x|} \quad \mbox{for almost every } x \in \mathbb{R}^d \, .
\end{equation}  
\end{enumerate}
\end{assumption}


\subsection{Bounds on the propagator of one-particle Schr\"odinger operators}\label{sec1}
In this section, we derive propagation bounds for one-particle Schr\"odinger operators with the
form $H_1$ as in (\ref{def:H_1}). To make a precise statement, we require the following 
from the potential $V$. 

\begin{assumption}[On $V$]  \label{Assump:V}  Let $V: \mathbb{R}^d \to \mathbb{C}$ have the form 
\beq \label{def:V}
V(x) =\int_{\R^d} \d \mu(k)\, e^{- i k\cdot x} \, 
\eeq
where $\mu : \text{Borel}(\mathbb{R}^d) \to \mathbb{R}$ is a Borel measure on $\mathbb{R}^d$ satisfying:
\begin{enumerate}[label=(\roman*)]
\item $\mu$ has support contained in a ball, i.e. there is some $M >0$ and $\supp\, \mu\subset B_M(0)$;
\item $\mu=\mu^+-\mu^-$ where $\mu^+$ and $\mu^-$ are non-negative finite measures on $\text{Borel}(\R^d)$, i.e. $\mu^{\pm}(\R^d)<\infty$. We set $|\mu|=\mu^++ \mu^-$.;
\item $\mu$ is even, i.e. $\mu(A) =\mu(-A)$ for all $A\in\text{Borel}(\R^d)$.
\end{enumerate}
\end{assumption} 

Under these assumptions, $V$ is the Fourier transform of a signed, compactly supported, finite measure $\mu$, which is real-valued and bounded.
Two parameters that will appear in estimates are $C_\mu$ and $M$, to characterize $V$. They need not be
chosen optimally but should satisfy
\be\label{Cmu+M}
\int_{\R^d} \d|\mu|(x) \leq C_\mu, \quad  \sup \{ |k| \mid k\in \supp( |\mu|) \}\leq M.
\ee

Two classes of examples of potentials $V$ satisfying Assumption \ref{Assump:V} are the following.

\noindent
(i) Let $V$ satisfy Assumption~\ref{Assump:V} and suppose the corresponding measure $\mu$ has a density $f\in L^1(\R^d)$. 
Then, these assumptions imply that $f$ has compact support, that $f(x)=f(-x)$, and that
\beq
V(x) = \int_{\R^d} \d k\, f(k) e^{-i k\cdot x}. 
\eeq
For example, our class of potentials $V$ includes $V(x)=\text{sinc}^k(x)$ for all $k\in\N$ for which the density is the $k$-fold convolution of the indicator function $f(y) =1_{[-1,1]}(y)$.

\noindent
(ii)  Let $V$ satisfy Assumption~\ref{Assump:V} and suppose the corresponding measure $\mu$ satisfies: 
There is some $N \in\N$, points $\{ a_n \}_{n=1}^N$ in $\mathbb{R}^d$, and numbers $\{ b_n \}_{n=1}^N$ in $\mathbb{R}$ 
for which 
\beq
\mu(A) =\frac 1 2 \sum_{n=1}^N b_n(\delta_{a_n}(A) + \delta_{-a_n}(A)) \quad \mbox{for any } A \in \text{Borel}(\mathbb{R}^d)  .
\eeq
Here $\delta_{(\cdot)}$ denotes the Dirac measure. This form gives rise to potentials $V$ with
\beq
V(x) = \sum_{n=1}^N  b_n \cos(a_n \cdot x) .
\eeq

The main result of this section is:

\begin{theorem}[Lieb-Robinson bound for Schr\"odinger operators]\label{Thm:boundProp}
Let $V$ satisfy Assumption~\ref{Assump:V} and consider the Schr\"odinger operator 
$H_1= -\Delta+V$ as defined in (\ref{def:H_1}). Then there exist constants $C_1, C_2,C_3>0$ depending on $d,\mu$, and $\sigma$ such that the estimate
\begin{equation} \label{ft_to_Gauss_est2}
\big|\< e^{-it H_1}f, \varphi^\sigma_x\>\big|
\leq C_1  e^{C_{2} |t| \ln |t|} \int_{\R^d} \d y\,e^{-\frac { C_{3} }{t^2+ 1}  |x-y|} |f(y)|
\end{equation}
holds for all $t\in\R$ and $f \in L^2( \mathbb{R}^d)$.
\end{theorem}

The constants $C_1,C_2$, and $C_3$ are derived in the proof of Corollary \ref{cor:LRB_SO_simple}.

\begin{remarks}\label{Rmk:LR}
\item\label{Rmk:LR_2subs}
Theorem \ref{Thm:boundProp} relies, in general, on the smoothness of the class of test functions $\varphi_y^\sigma$ that we used to probe the locality properties of the dynamics. For example, we can see from an explicit computation in the case of $V=0$, that the exponential decay does generally not hold when the Gaussians is 
replaced by a non-smooth functions such as, for example a characteristic function. Using the formula 
\beq 
(e^{-i t (-\Delta)} \psi)(x) = \frac 1 {(4 \pi i t)^{d/2}} \int_{\R^d} \d y\, e^{\frac{i |x-y|^2}{4 t}} \psi(y),
\eeq
valid for $t\neq 0$ and general $\psi\in L^1(\R^d)\cap L^2(\R^d)$, a straightforward calculation, in the case $d=1$, shows that the leading behavior of
$|(e^{-i t (-\Delta)} \chi_{[-1,1]})(x)|$, for $t>0$ and $|x|$ large is given by
$$
|(e^{-i t (-\Delta)} \chi_{[-1,1]})(x)| \sim 2\sqrt{\frac{t}{\pi}} \frac{x}{x^2 -1} \sin \frac{x}{2t}.
$$
\item In the case $V=0$ and for the Gaussians $\varphi_y^\sigma$  another explicit computation (see e.g. \cite[Sec. 7.3]{teschl:2009}), shows that 
\beq \label{Gauss_bd}
\big|\big( e^{-i t (-\Delta)}\varphi_y^\sigma\big)(x)\big| =
\frac 1{(2\pi)^{d/2}}  \frac {e^{-\frac{\sigma ^2 |x-y|^2}{8t^2+2\sigma^4}}}{(4 t^2 + \sigma^4)^{d/4}}.
\eeq
Using this form one immediately sees that an estimate analogous to \eq{ft_to_Gauss_est} holds with a Gaussian distance dependence of the kernel.
%
%
\item
In view of the first two remarks it is clear that Theorem \ref{Thm:boundProp} is far from optimal. It is also likely that similar bounds hold for a broader class of 
potentials $V$. 
\end{remarks}

\subsection{Lieb-Robinson bounds and the thermodynamic limit}\label{sec:LRB+TL}

Our next results concern Lieb-Robinson bounds for the Heisenberg dynamics associated to 
the operator $H_{\Lambda}^{\sigma}$ defined in (\ref{Def:H_L_sig}). We begin by recalling the notion of
dynamics on Fock space. 

As before, let $\mathcal{B}(\mathfrak{F}^-)$ denote the bounded linear operators on  $\mathfrak{F}^-$.
For each $\sigma>0$ and any bounded, measurable $\Lambda \subset \mathbb{R}^d$, we define
the Heisenberg dynamics associated to $H_{\Lambda}^{\sigma}$ for each $t \in \mathbb{R}$ as
the map $\tau_t^{\Lambda} : \mathcal{B}( \mathfrak{F}^-) \to \mathcal{B}( \mathfrak{F}^-)$ defined by
\beq \label{fv_pert_dyn}
\tau_t^\Lambda (A) =  e^{ it H_\Lambda^{\sigma}} A e^{-it H_\Lambda^{\sigma}} \quad \mbox{for all } A \in \mathcal{B}( \mathfrak{F}^-) \, .
\eeq
We note that although $\tau_t^\Lambda (A)$ depends on $\sigma$, we have suppressed this in our notation. 

We will analyze this dynamics on the CAR algebra generated by the set 
$\big\{ a(f), a^*(f): f\in L^2(\R^d)\big\}$. Again, we refer to \cite[Sec 5.2]{bratteli:1997} for more details. 
In particular, we will focus our attention to operators $A= a(f) \in \mathcal{B}( \mathfrak{F}^-)$ for some $f\in L^2(\R^d)$. 

Let us also recall the free dynamics, i.e. the case where $W$ =0 and there is no interaction. 
We will denote this well-studied, free dynamics of the CAR algebra by
\beq
\tau_t^\emptyset(a(f)) = e^{i t \d\Gamma(H_1)} a(f) e^{-i t \d\Gamma(H_1)} ,\qquad  f\in L^2(\Rl^d),\ \  t\in \Rl \, .
\eeq
A straight forward calculation shows
\be
\tau_t^\emptyset(a(f)) = a(f_t)\qquad \text{where}\qquad f_t = e^{-it H_1} f.
\ee
Our goal is to examine the behavior of $\tau_t^{\Lambda}$ as $\Lambda$ tends to $\mathbb{R}^d$, and in particular,
we wish to establish the existence of a dynamics in this thermodynamic limit. To do so, we regard 
$\tau_t^{\Lambda}$ as a perturbation of the infinite volume free dynamics $\tau_t^{\emptyset}$ on the
finite volume $\Lambda$. In this case, the key to constructing the thermodynamic limit is an appropriate
form of the Lieb-Robinson bound. To express the Lieb-Robinson bound for this model, we find it
convenient to introduce the non-negative function 
\begin{equation} \label{Def:F_t^L}
F_t^{\Lambda}(f,g) = \Vert \{ \tau^\Lambda_t(a(f)), a^*(g)\} - \{ \tau^\emptyset_t(a(f)), a^*(g)\} \Vert + \Vert \{ \tau^\Lambda_t(a(f)), a(g)\}\Vert
\end{equation}
Iteration is at the heart of most Lieb-Robinson bounds, and in the present context, our proof will show that 
the function $F_t^{\Lambda}$ above iterates more simply than either term on the right-hand-side of (\ref{Def:F_t^L}).
In any case, we find the following Lieb-Robinson bound.

\begin{theorem}[Lieb-Robinson bound]\label{thm:propagationbound}
Fix $\sigma>0$. Let $V$ satisfy Assumption~\ref{Assump:V}, $W$ satisfy Assumption~\ref{Assump:W}, and for each $t \in \mathbb{R}$ and 
any bounded, measurable set $\Lambda \subset \mathbb{R}^d$, denote by $\tau_t^{\Lambda}$ the dynamics
associated to $H_{\Lambda}^\sigma$ as defined in (\ref{fv_pert_dyn}). For any $f,g \in L^1( \mathbb{R}^d) \cap L^2( \mathbb{R}^d)$, the bound
\begin{equation}
F_t^{\Lambda}(f,g) \leq D(t) (e^{P_3(t)}-1)  \int_{\R^d} \int_{\R^d}  \d x \d y\, e^{-\frac{c_t|x-y|}{4}} |f(x)||g(y)|
\end{equation} 
holds for functions $D(t) \sim e^{c|t| | \ln |t| |}$, $P_3$ a polynomial of degree $6d+1$ in $|t|$, and $c_t \sim \frac{1}{1+t^2}$.  
Explicit values for these functions are given in Section~\ref{sec:loc_bds}, see specifically Lemma~\ref{sec2:lm:boundK2} and Lemma~\ref{lm:boundafg}.
\end{theorem}

\begin{remark}
Here it is crucial to subtract the free time evolution 
\beq
\big\{\tau_t^\emptyset(a^*(f)), a(g)\big\} = \<f_t,g\> \idty
\eeq
since $|\<f_t,g\> |$ does not, in general, decay exponentially; see Remark \ref{Rmk:LR} (i). 
\end{remark}
%
%

Our main application concerns the existence of a dynamics in the thermodynamic limit.
In Section~\ref{sec:thermlimit:proof} we show how the following theorem is a consequence of the $\Lambda$-independent bounds
proven in Theorem~\ref{thm:propagationbound}.

\begin{theorem}\label{thm:main}
Under the assumptions of Theorem  \ref{thm:propagationbound}
there exists a strongly continuous one-parameter group of automorphisms of the CAR algebra over $L^2(\Rl^d)$, $\{\tau_t\}_{t\in\Rl}$,
such that for all $f\in L^2(\Rl^d)$ and any increasing sequence $(\Lambda_n)$ of bounded subsets of $\Rl^d$ such that $\cup_{n}\Lambda_n=\Rl^d$,
\be\label{TLmonomials}
\lim_{n\to \infty} \tau^{\Lambda_n}_t(a(f)) = \tau_t(a(f))
\ee
in the operator norm topology, with convergence uniform in $t$ in compact subsets of $\Rl$.
\end{theorem}

\section{Lieb-Robinson Bound for Schr\"odinger Operators. Proof of Theorem \ref{Thm:boundProp}}\label{sec:boundProp:pf}

In this section we use the notation $H_0=-\Delta$ and $H_1=-\Delta +V$.
To prove Theorem \ref{Thm:boundProp} we will use a Dyson series expansion for $e^{itH_1}$:
\be\label{dyson}
e^{-i t H_1 }=e^{-it H_0}+ (-i)^n\sum_{n=1}^\infty \int_0^{t}\d t_{n} \cdots \int_0^{t_2} \d t_1\, 
e^{-i (t-t_{n}) H_0} V e^{-i (t_{n}-t_{n-1} )H_0 }V \cdots V e^{i t_1H_0}.
\ee
Since $V$ is bounded (by Assumption \ref{Assump:V}), this  series is absolutely convergent in norm. We are interested
in estimating $\big|\big(e^{-i t H_1 }\varphi_y^\sigma\big)(x)\big|$, where $\varphi_y^\sigma$ is the Gaussian function
given in \eq{Def:Gaussian}. Using the Fourier representation of $V$ \eq{def:V}, the integrand of the n-th term in the expansion \eq{dyson} applied to 
$\varphi_y^\sigma$ can be expressed as follows:
\begin{align}
&\big(e^{-i (t-t_{n}) H_0} V e^{-i (t_{n}-t_{n-1} )H_0 }V \cdots V e^{i t_1H_0} \varphi_y^\sigma\big)(x)\\
& = 
\int_{\R^d} \d\mu(k_n) \cdots \int_{\R^d} \d\mu(k_1)A(t_1,...,t_{n},t,k_1,...,k_n,y, x).\notag
\end{align}
where
\begin{align}
A(t_1,...,t_{n},t,k_1,...,k_n,y, x) = \big(
e^{-i (t-t_{n}) H_0} V_{k_n} e^{-i (t_{n}-t_{n-1} )H_0 }V_{k_{n-1}} \cdots V_{k_1} e^{i t_1H_0} \varphi_y^\sigma
\big)(x).
\end{align}
Here $V_k$ is the multiplication operator by $V_k(x)= e^{-ik \cdot x}$ and $|\mu|= \mu^++\mu^-$.  

\begin{lemma}\label{lm:Fourier}
Let $n\in\N$. For all $k_1,...,k_n,y,x\in\R^d$, $t\geq t_n\geq ...\geq t_{1} \geq 0$ one has
\begin{align}\label{lm:Fourier:eq1}
\big|A(t_1,...,t_n,t,k_1,...,k_n,y,x)\big| = \frac 1{(2\pi)^{d/2}}
\frac {e^{-\frac{\sigma^2}{8t^2+2\sigma^4}\big|(x-y)-2\sum_{l=0}^{n}(t_{l+1}-t_{l})\sum_{j=1}^{l} k_j\big|^2}} {(4t^2+\sigma^4)^{d/4} },
\end{align}
where we use the conventions $t_{n+1}=t$, $t_{0}=0$ and $\sum_{j=1}^0 =0$. Furthermore, 
\beq
\big|\big( e^{-i t H_0}\varphi_y^\sigma\big)(x)\big| =
\frac 1{(2\pi)^{d/2}}  \frac {e^{-\frac{\sigma ^2 |x-y|^2}{8t^2+2\sigma^4}}}{(4 t^2 + \sigma^4)^{d/4}}.\label{lm:Fourier:eq2}
\eeq
\end{lemma}

\begin{proof}
Let $\mathcal F$ be the unitary Fourier transform on $\R^d$ and $\mathcal F^*$ be its inverse then we obtain
\be
A(t_1,...t_n,t,k_1,...,k_n,y,x) 
=
\big(\mathcal F^*\mathcal F e^{-i (t-t_{n}) H_0}\mathcal F^*\mathcal F V_{k_n}\mathcal F^*  \cdots \mathcal F V_{k_1} \mathcal F^*\mathcal F e^{i t_1 H_0} \mathcal F^*\mathcal F \varphi_y^\sigma\big)(x).\label{eq:FOurier}
\ee
Now, for all $t\in\R$, $k\in\R^d$ and $\psi\in L^2(\R^d)$ we have $\mathcal F e^{-i t H_0}\mathcal F^*\psi = e^{-i (\cdot)^2 t }\psi $ and  $\mathcal F V_k\mathcal F^*\psi =\psi(\cdot- k)$. 
Therefore, 
$$
A(t_1,...t_n,t,k_1,...,k_n,y,x) =
\frac 1 {(2\pi)^{d/2}} \int_{\R^d} \d k\,  e^{i kx} \prod_{l=0}^n e^{-i(t_{l+1}-t_{l}) \big(k- \sum_{j=l+1}^{n} k_j \big)^2}
\big(\mathcal F\varphi_y^\sigma\big)(k- \sum_{j=1}^n k_j),
$$
where we use the convention that $\sum_{j=n+1}^n =0$. 
Next we use that
$$
(\mathcal F \varphi_y^\sigma)(k) = \frac 1 {(2\pi)^{d/2}} e^{-i k\cdot y}e^{-\frac{\sigma^2 |k|^2} 2}
$$
and perform a change of variables to obtain
\be
A(t_1,...t_n,t,k_1,...,k_n,y,x) =
\frac {e^{i \sum_{j=1}^n k_j \cdot(x-y)}} {(2\pi)^{d}} \int_{\R^d} \d k\,  e^{i k\cdot(x-y)} \prod_{l=0}^n e^{-i(t_{l+1}-t_{l}) \big|k + \sum_{j=1}^l k_j \big|^2}
 e^{-\frac{\sigma^2 |k|^2} 2},\label{eq:FOurier3}
\ee
where we use the convention $\sum_{j=1}^0 =0$.
Multiplying out 
$\big| k + \sum_{j=1}^l k_j \big|^2= |k|^2 + 2 k \cdot\sum_{j=1}^l k_j + \big|\sum_{j=1}^l k_j\big|^2$, using $\sum_{l=0}^n(t_{l+1}-t_{l})=t_{n+1}-t_0=t$ 
and taking the absolute value give
$$
\big|A(t_1,...t_n,t,k_1,...,k_n,y,x) \big|
=
\frac 1{(2\pi)^{d}} \Big| \int_{\R^d} \d k\,  e^{i k\cdot(x-y)} e^{-i |k|^2 t} e^{-2i \sum_{l=0}^n(t_{l+1}-t_{l})  \sum_{j=1}^l k\cdot k_j} e^{-\frac{\sigma^2 |k|^2} 2}\Big|.
$$
Now, a calculation shows that for any $c>0$ and $a,b\in\R$
\beq
\frac 1 {(2\pi)^{1/2}}\Big|
\int_{\R} \d k\,  e^{ik b} e^{-k^2(c+i a)} 
\Big| = \displaystyle
\frac{e^{-\frac{c b^2}{4(a^2+c^2)}}}{(4(a^2+c^2))^{1/4}}.
\eeq
Hence,
\begin{align}
\big|A(t_1,...t_n,t,k_1,...,k_n,y,x) \big|=
\frac 1{(2\pi)^{d/2}}
\frac 1 {(4 t^2+\sigma^4)^{d/4} } 
e^{-\frac{\sigma^2}{8t^2+2\sigma^4}\big|(x-y)-2\sum_{l=0}^n(t_{l+1}-t_{l})\sum_{j=1}^l k_j\big|^2},
\end{align}
which proves \eqref{lm:Fourier:eq1}. Identity \eqref{lm:Fourier:eq2} follows from an explicit calculation using the integral kernel of $e^{-i t H_0}$, see e.g. \cite[Sec. 7.3]{teschl:2009}. 
\end{proof}

With this lemma we have arrived at the following estimate. For $t\in\Rl$ and $x,y\in\R^d$, we have
\begin{align}\label{eq:Cor:mainrsult}
&\big|\big( e^{-i t H_1} -  e^{-i t H_0}\big)\varphi_y^\sigma(x)\big|
\leq \frac{1}{(2\pi)^{d/2}}\frac{1} {(4 t^2+\sigma^4)^{d/4} }  \sum_{n=1}^\infty \int_{\R^d} \d |\mu|(k_1) \cdots \int_{\R^d} \d |\mu|(k_n)\\
&\times
\int_0^{|t|}\d t_{n} \cdots \int_0^{t_2} \d t_1\,   
e^{-\frac{\sigma^2}{8t^2+2\sigma^4}\big|(x-y)-2\sum_{l=0}^n(t_{l+1}-t_{l})\sum_{j=1}^l k_j\big|^2}\Big),\notag
\end{align}
where, as before, we use the convention $\sum_{j=1}^0 = 0 $. 
By estimating the RHS of this estimate, we will obtain the following proposition. Recall the definitions of $C_\mu$ and $M$ in \eq{Cmu+M}.

\begin{proposition}\label{prop:LRB_SO}
For all $t\in\Rl$ and $x,y\in\R^d$, we have
$$
\big|\big( e^{-i t H_1} -  e^{-i t H_0}\big)\varphi_y^\sigma(x)\big|
\leq \frac{1}{(2\pi)^{d/2}}\frac{1} {(4 t^2+\sigma^4)^{d/4} } 
\big(  e^{-\frac{\sigma^2 |x-y|^2}{32t^2+8\sigma^4}}(e^{C_\mu |t|}-1)
+\frac 1 {\sqrt{2\pi}} e^{-\frac{|x-y|}{4 |t| M}(  \ln \frac{|x-y|} { 4 MC_\mu t^2}-1)}e^{C_\mu |t|}
\big).
$$
\end{proposition}

Using \eq{lm:Fourier:eq2} to bound  $| e^{-i t H_0} \varphi_y^\sigma(x)|$, it is then straightforward to obtain an estimate for $| e^{-i t H_1} \varphi_y^\sigma(x)|$.
In our application, however, the following simplified estimate is easier to use.

\begin{corollary}\label{cor:LRB_SO_simple}
There exist constants $C_1,C_2$, and $C_3$, such that for all $t\in\Rl$ and $x,y\in\R^d$, we have
\be
\big| e^{-i t H_1} \varphi_y^\sigma(x)\big|\leq C_1 e^{C_2  |t| \ln |t| - C_3 \frac{|x-y|}{1+t^2}}.
\ee
\end{corollary}

\begin{proof}[Proof of Proposition \ref{prop:LRB_SO}]
Note that, as expected from translation invariance, the RHS of \eq{eq:Cor:mainrsult} depends only on $x-y$. Therefore, w.l.o.g., we can assume $y=0$.
Let $I_n$, $n\geq 1$, denote the n-th term of the sum in the 
RHS of \eq{eq:Cor:mainrsult}:
\be\label{In}
I_n=  \int_{\R^d} \d |\mu|(k_1) \cdots \int_{\R^d} \d |\mu|(k_n)
\int_0^{|t|}\d t_{n} \cdots \int_0^{t_2} \d t_1\,   
e^{-\frac{\sigma^2}{8t^2+2\sigma^4}\big|(x-y)-2\sum_{l=0}^n(t_{l+1}-t_{l})\sum_{j=1}^l k_j\big|^2}.
\ee
Given $x\in\R^d$, $t\in \Rl$, 
we split the sum over $n$ in \eq{eq:Cor:mainrsult} as follows:
\be
B_0=\sum_{n=1}^{N_0} I_n, \quad B_1=\sum_{n=N_0+1}^\infty  I_n, \mbox{ with } N_0 = \left[\frac{|x|}{(4|t| M)}\right],
\ee
where $[a]$ denotes the integer part of $a$ and $M$ is as in \eq{Cmu+M}.

First, if $\frac{|x|}{(4|t| M)}\geq 1$, the first sum is non-empty and we estimate its terms as follows. Let $n\leq |x|/(4|t|M)$, and
note that for $k_j\in B_M(0)$ with $1 \leq j \leq l \leq n$ and $|t|=t_{n+1}\geq t_{n}\geq...\geq t_{0}=0$ 
\beq\label{eq6}
\Big|\sum_{l=0}^n(t_{l+1}-t_{l})\sum_{j=1}^l k_j\Big|\leq M n \sum_{l=0}^n(t_{l+1}-t_{l})  = M n |t|. 
\eeq
Therefore,
\beq
\big|x-2\sum_{l=0}^n(t_{l+1}-t_{l})\sum_{j=1}^l k_j\big|^2\geq \frac{1}{4}|x|^2.
\eeq
and therefore
\be\label{B0}
B_0 \leq e^{-\frac{\sigma^2 |x|^2}{32t^2+8\sigma^4}} \sum_{n=1}^{N_0}    \frac{C_\mu^n |t|^n}{n!}\leq   e^{-\frac{\sigma^2 |x|^2}{32t^2+8\sigma^4}}(e^{C_\mu |t|}-1)
\ee
To estimate the terms in $B_1$, note that the integrand in \eq{In} is bounded by 1.
Therefore, using $C_\mu$ defined in  \eq{Cmu+M}, 
\be
B_1 \leq 
 \sum_{n=N_0+1}^\infty \frac{(C_\mu |t|)^n}{n!} \leq \frac {(C_\mu |t|)^{N_0+1}} {(N_0+1)!} e^{C_\mu |t|}.
\ee
Stirling's formula yields for all $m\geq 1$ the bound
\beq
\frac 1 {m!}\leq \frac 1 {\sqrt{2\pi}} e^{m-m \ln m}.
\eeq
Using this and $N_0+1 \geq |x|/(4 |t| M)$, we obtain 
\be \label{B1}
B_1 \leq \frac 1 {\sqrt{2\pi}} e^{-\frac{|x|}{4 |t| M}(  \ln \frac{|x|} { 4 MC_\mu t^2}-1)}e^{C_\mu |t|}.
\ee

If $\frac{|x|}{(4|t| M)}< 1$, $B_0 =0$ and we estimate $B_1$ as in \eq{B1} with $N_0=0$.

The proposition is proved by combining the estimates \eq{B0} and \eq{B1}.
\end{proof}

\begin{proof}[Proof of Corollary \ref{cor:LRB_SO_simple}]
We may again restrict ourselves to the case $y=0$. Proposition \ref{prop:LRB_SO} and \eq{lm:Fourier:eq2} together immediately give the estimate
\begin{align}\label{start_cor}
\big| e^{-i t H_1} \varphi_0^\sigma(x)\big| \leq & \frac{1}{(2\pi)^{d/2}} \frac{1}{(4t^2+\sigma^4)^{d/4}} \notag\\
&\times\left[ 
e^{-\frac{\sigma^2 |x|^2}{8t^2+ 2\sigma^4}} + e^{-\frac{\sigma^2 |x|^2}{32t^2+ 8\sigma^4}} (e^{C_\mu |t|} -1) 
+\frac{1}{\sqrt{2\pi}} e^{-\frac{|x|^2}{4M |t|} (\ln \frac{|x|}{4 M C_\mu t^2} -1)} e^{C_\mu |t|}
\right]. \notag
\end{align}
The last term in the square bracket is the estimate \eq{B1} for $B_1$, which we can simplify by considering two cases for $(x,t)\in \R^{d+1}$,
namely $\frac{|x|}{4 M C_\mu t^2 } \geq e^{2}$, and $\frac{|x|}{4 M C_\mu t^2 } < e^{2}$.
In the first case, we have
\beq
B_1\leq \frac {e^{-\frac{|x|}{4 |t| M} + C_\mu |t| }}{\sqrt{2\pi}}.
\label{sec:boundeq2111}
\eeq
On the other hand, if $\frac{|x|}{4 M C_\mu t^2 } < e^{2}$,
we use the inequality $e^{-u \ln u}\leq e^{\frac{1 } e}$ for all $u>0$, to obtain 
\be
B_1 \leq  \frac{1}{\sqrt{2\pi}} e^{ \frac{1}{e}} e^{e^2 C_\mu |t| (\ln C_\mu  |t| +1)+C_\mu |t|} 
e^{-\frac{|x|}{4 M C_\mu t^2 }+ e^2}.
\label{sec:boundeq211}
\ee
By bounding $B_1$ by the sum of the RHSs of \eq{sec:boundeq2111} and  \eq{sec:boundeq211} and making a few more 
easy simplifications we arrive at the following bound:
\begin{align}
\big| e^{-i t H_1} \varphi_0^\sigma(x)\big| \leq & \frac{1}{(2\pi\sigma^2)^{d/2}}\notag\\
&\times\left[ 
e^{-\frac{\sigma^2 |x|^2}{32t^2+ 8\sigma^4}} e^{C_\mu |t|}  + \frac{1}{\sqrt{2\pi}}  e^{-\frac{|x|}{4 |t| M} + C_\mu |t| }
+ \frac{1}{\sqrt{2\pi}} e^{ \frac{1}{e}} e^{e^2 C_\mu |t| (\ln C_\mu  |t| +1)+C_\mu |t|} 
e^{-\frac{|x|}{4 M C_\mu t^2 }+ e^2}
\right]. \notag
\end{align}

To estimate the Gaussian decay of the first term between the square brackets by a simple exponential, we use that for all $u \in \Rl$, $u^2 \geq u -1/4$. Furthermore, since $u\ln u\geq u-1$, we also have
$$
e^{C_\mu |t|} \leq  e^{ 1 + C_\mu |t| \ln C_\mu |t| }.
$$
Using this and replacing the constant prefactors by their maximum, we find
$$
\big| e^{-i t H_1} \varphi_0^\sigma(x)\big| \leq \frac{1}{(2\pi\sigma^2)^{d/2}}
e^{\frac{1}{8\sigma^2}+ \frac{1} e + e^2} e^{e^2 C_\mu |t| (\ln C_\mu |t| +2) }   \Big( e^{-\frac{\sigma^2 |x|}{32t^2+8\sigma^4}}   +   e^{-\frac{|x|}{4 t M}}  +  e^{-\frac{|x|}{4 M C_\mu t^2} } \Big).
$$
Finally, we use the estimate
\begin{align}
\min\Big\{\frac{\sigma^2}{32t^2+8\sigma^4} , \frac{1}{4 t M} , \frac{1}{4 M C_\mu t^2} \Big\}
&\geq \frac {1}{t^2( \frac{32}{\sigma^2} + 4 M C_\mu) + 4 t M+ 8 \sigma^2  }\notag\\
&\geq
\frac {1}{t^2( \frac{32}{\sigma^2}+ 4 M C_\mu+ 2) + 2 M^2 + 8 \sigma^2 }\notag\\
&\geq \frac{1}{\frac{32}{\sigma^2}+ 4 M C_\mu+ 2 (M^2 +1) + 8 \sigma^2 }  \frac{1}{t^2 +1}
\end{align}
to bound the sum of three exponentials by 
$$
3 e^{-  C_3 \frac {|x|}{t^2+1}}, \mbox{ with } C_3 =  \frac{1}{\frac{32}{\sigma^2}+ 4 M C_\mu+ 2 (M^2 +1) + 8 \sigma^2 }.
$$
It is now straightforward to find suitable values for $C_1$ and $C_2$ for which the bound given in the corollary holds.
\end{proof}

\section{Many-body Lieb-Robinson bound. Proof of Theorem \ref{thm:propagationbound}} \label{sec:loc_bds}

The main goal of this section is to prove Theorem~\ref{thm:propagationbound}.
We do so in three steps. First, in Section~\ref{sec:pre_bd} we establish a basic estimate which facilitates an
iteration scheme; this is the content of Lemma~\ref{lm:BoundAntiComm}. Next, in Section~\ref{sec:est_K_t}, we estimate a kernel function
which, among other things, ultimately justifies the convergence of our iteration, see Lemma~\ref{lm2.1} and Lemma~\ref{sec2:lm:boundK2}.
Finally, in Section~\ref{sec:it_scheme} we perform the iteration and verify the bound claimed in Theorem~\ref{thm:propagationbound}. 

\subsection{A Preliminary Bound} \label{sec:pre_bd}
In this section, we provide an estimate on the basic quantity of interest in Theorem~\ref{thm:propagationbound}.
Let us briefly recall the set-up. We have fixed $\sigma>0$, taken $V$ and $W$ satisfying Assumption~\ref{Assump:V} and 
Assumption~\ref{Assump:W} respectively, and introduced, see (\ref{Def:F_t^L}), the non-negative function
\begin{equation} \label{Def:F_t_2}
F^{\Lambda}_t(f,g) = \big\|  \{ \tau_t^{\Lambda}( a( f ) ), a^*( g ) \} - \{ \tau_t^\emptyset( a( f ) ), a^*( g ) \} \big\| + \left\| \{ \tau_t^{\Lambda}( a^*( f) ), a^*( g) \} \right\|
\end{equation}
for any bounded, measurable set $\Lambda \subset \mathbb{R}^d$, $t \in \mathbb{R}$, and functions $f,g \in L^2( \mathbb{R}^d)$. 
Our first estimate is as follows. 
\begin{lemma}\label{lm:BoundAntiComm}
Under the assumptions described above, for any $t\geq0$, we find that
\begin{eqnarray} \label{F_it_bd}
F^{\Lambda}_t(f,g) & \leq &  C_\sigma\int_0^t \d s\, \int_{\R^d} \d x\,K_{t-s}(f, x)  |\< e^{-i sH_1}\varphi_x^\sigma , g \>| \nonumber \\
& \mbox{ } & \quad + C_\sigma\int_0^t \d s\, \int_{\R^d} \d x\,K_{t-s}(f, x) F^{ \Lambda}_s( \varphi^{\sigma}_x, g)  
\end{eqnarray}
where $C_{\sigma} >0$ is as in (\ref{Def:C_sig}) and with kernel function $K_t(f,x)$ given by
\begin{equation} \label{Def:K}
 K_{t}(f,x) = \|W \|_1| \langle e^{-i tH_1}f, \varphi^{\sigma}_x \rangle |  + 2 \left( |W| * | \langle e^{-i tH_1}f, \varphi^{\sigma}_{(\cdot)} \rangle | \right)(x) \, . 
\end{equation}
\end{lemma}

\begin{proof}
We begin by recalling a useful perturbation formula. 
Fix a bounded, measurable set $\Lambda \subset \mathbb{R}^{d}$ and take
$\sigma >0$, $ t \geq 0$, and $f \in L^2( \mathbb{R}^{d})$. In this case, 
\begin{equation}\label{perturbation_formula}
\tau_t^{\Lambda}( a( f ) )  =   \tau_t^\emptyset( a( f ) )  + i \int_0^t \d s\,  \tau_s^{\Lambda} \left( \left[ W_{\Lambda}^{\sigma}, \tau_{t-s}^\emptyset(a( f )) \right] \right) \, 
\end{equation}
a proof of which can be found in \cite[Prop. 5.4.1]{bratteli:1997}. Note that
\begin{align}
\left[ W_{\Lambda}^{\sigma}, \tau_{t-s}^\emptyset(a(f)) \right]  
 = &\frac 1 2 \int_{\mathbb{R}^{d}}  \int_{\mathbb{R}^{d}} \d x \d y \,W_\Lambda(x,y) \left[a^*( \varphi^{\sigma}_x) a^*( \varphi^{\sigma}_y) a( \varphi^{\sigma}_y) a( \varphi^{\sigma}_x) , a(f_{t-s}) \right]    \nonumber \\
 = &\frac 1 2 \int_{\mathbb{R}^{d}} \int_{\mathbb{R}^{d}} \d x  \d y\, W_\Lambda(x,y) \left[a^*( \varphi^{\sigma}_x) a^*( \varphi^{\sigma}_y), a(f_{t-s}) \right] a( \varphi^{\sigma}_y) a( \varphi^{\sigma}_x)  
\end{align}
where we have set $f_t = e^{-it H_1}f$ and used \eqref{CAR}. Further calculating, we find
\begin{align}
\left[ a^*(\varphi^{\sigma}_x) a^*(\varphi^{\sigma}_y), a( f_{t-s}) \right] 
= & a^*(\varphi^{\sigma}_x) \left\{ a^*(\varphi^{\sigma}_y), a( f_{t-s})  \right\}  -  \left\{ a^*(\varphi^{\sigma}_x), a( f_{t-s})  \right\}  a^*(\varphi^{\sigma}_y) \nonumber \\
 = &  \langle f_{t-s}, \varphi^{\sigma}_y  \rangle a^*(\varphi^{\sigma}_x) -\langle f_{t-s}, \varphi^{\sigma}_x \rangle a^*(\varphi^{\sigma}_y).
\end{align}
Using now the symmetry of $W$, we obtain 
\begin{align*}
&\{ \tau_t^{\Lambda}( a( f ) ), a^*( g ) \} - \{ \tau_t^\emptyset( a( f ) ), a^*( g ) \} \notag\\
=&     i \int_0^t \d s \int_{\mathbb{R}^{d}} \int_{\mathbb{R}^{d}} \d x \d y\, W_\Lambda(x,y)   
  \langle f_{t-s}, \varphi^{\sigma}_y  \rangle\left\{ \tau_s^{ \Lambda} \left( a^*(\varphi^{\sigma}_x) \right) \tau_s^{\Lambda} \left( a(\varphi^{\sigma}_y) \right) \tau_s^{\Lambda} \left( a(\varphi^{\sigma}_x) \right) , a^*(g) \right\}. 
\end{align*}
With the anti-commutator relation 
\begin{equation}
\{ ABC, D \} = \{A,D\} BC - A \{B,D\}C + AB \{C,D\} \, ,
\end{equation}
the norm bound
\begin{align} \label{it_est_1}
&\left\|  \{ \tau_t^{\Lambda}( a( f ) ), a^*( g ) \} - \{ \tau_t^\emptyset( a( f ) ), a^*( g ) \} \right\| \notag\\
 \leq &
  C_{\sigma} \int_0^{t} \d s \int_{\R^d} \int_{\R^d} \d x \d y |W(x-y)|   | \langle f_{t-s}, \varphi^{\sigma}_y  \rangle |  \nonumber\\
&\qquad \times\left( \left\| \left\{ \tau_s^{\Lambda}(a^*(\varphi^{\sigma}_x)), a^*(g) \right\} \right\|  + 
\left\| \left\{ \tau_s^{\Lambda}(a(\varphi^{\sigma}_y)), a^*(g) \right\} \right\|  \right. 
  \left. +  \left\| \left\{ \tau_s^{\Lambda}(a(\varphi^{\sigma}_x)), a^*(g) \right\} \right\| \right) 
\end{align}
readily follows. Here we have used the bound $|W_\Lambda(x,y)|\leq |W(x-y)|$ for all $x, y \in\R^d$. Similar arguments yield the bound
\begin{align} \label{it_est_2}
&\left\| \{ \tau_t^{ \Lambda}( a^*( f) ), a^*( g) \} \right\|  \notag\\
\leq  &
C_{\sigma} \int_0^{t} \d s\int_{\R^d} \int_{\R^d} \d x \d y\, |W(x-y)| | \langle  f_{t-s}, \varphi^{\sigma}_y \rangle |  \notag  \\
 &\quad\quad\times \left( \left\| \left\{ \tau_s^{\Lambda}(a^*(\varphi^{\sigma}_x)), a^*(g) \right\} \right\|   + \left\| \left\{ \tau_s^{\Lambda}(a^*(\varphi^{\sigma}_y)), a^*(g) \right\} \right\|  +  \left\| \left\{ \tau_s^{\Lambda}(a(\varphi^{\sigma}_x)), a^*(g) \right\} \right\| \right) .
\end{align}

Our goal, as in most Lieb-Robinson bounds, is to derive bounds which can be iterated. Since neither (\ref{it_est_1}) nor (\ref{it_est_2}) iterate separately, 
we bound their sum, i.e. the function $F_t^{\Lambda}(f,g)$ introduced in (\ref{Def:F_t_2}) above. Recalling that 
$\{ \tau_t^\emptyset( a( f) ), a^*( g) \}   =   \langle f_t, g \rangle \idty$, we find
\begin{align} \label{eq:Fbd}
F^{ \Lambda}_t(f,g) & \leq  C_\sigma \int_0^t \d s\int_{\R^d} \int_{\R^d} \d x\d y\, |W(x-y)| | \langle f_{t-s},  \varphi^{\sigma}_y \rangle |  \notag\\
&\qquad \qquad \qquad\times\left(2F^{ \Lambda}_s( \varphi^{\sigma}_x, g) + 2 |\langle e^{-i sH_1}\varphi_x^\sigma , g \rangle| + F^{ \Lambda} _s(\varphi^{\sigma}_y,g) + |\langle e^{-i sH_1} \varphi_y^\sigma , g \rangle|  \right)  \nonumber \\
& =  C_\sigma \int_0^t \d s\,\int_{\R^d} \d x\,K_{t-s}(f, x) \big(F^{ \Lambda}_s( \varphi^{\sigma}_x, g) + |\langle e^{-i sH_1}\varphi_x^\sigma , g \rangle|\big),
\end{align}
where we have introduced $K_t(f,x)$ as in \eqref{Def:K}. This is the claim in (\ref{F_it_bd}). 
\end{proof}

Looking at the bound proven in Lemma~\ref{lm:BoundAntiComm}, in particular (\ref{F_it_bd}), it is natural to begin an iteration. 
To ensure convergence, we first provide an estimate on 
the kernel function $K_t(f,x)$. 

\subsection{Estimating the Kernel} \label{sec:est_K_t}
In this section, we provide two useful apriori estimates on the kernel function $K_t(f,x)$ defined in (\ref{Def:K}) 
of Lemma~\ref{lm:BoundAntiComm}.
First, we estimate this function for general $f \in L^2( \mathbb{R}^d)$, this result is stated in
Lemma~\ref{lm2.1} below. In Lemma~\ref{sec2:lm:boundK2}, we provide a similar estimate in the special case that
$f$ is a Gaussian.

Before we prove our first estimate, it will be convenient to introduce the following notation.
For any $r>0$, set $G_r: \mathbb{R}^d \to \mathbb{R}$ to be
\begin{equation} \label{exp_dec}
G_r(x) = e^{-r|x|} \quad \mbox{for all } x \in \mathbb{R}^d \, .
\end{equation}

\begin{lemma}\label{lm2.1}
Let $t \geq 0$, $f\in L^2(\R^d)$, and $x\in\R^d$. We have 
\beq \label{kern_bd}
K_t(f,x)\leq P_1(t) e^{C_{2} t|\ln t|} \big(G_{b_t}*|f|\big)(x)
\eeq
where $G_{b_t}$ is as in (\ref{exp_dec}) with 
\begin{equation} \label{Def:bt}
4 b_t = \frac{aC_3}{  a(t^2+1) + C_3 }
\end{equation} 
and $P_1: \mathbb{R} \to (0, \infty)$ is the polynomial of degree 2d given by 
\begin{equation} 
P_1(t)=C_1(2c_WD_3^2 (4b_t)^{-d}+\|W\|_1) \, .
\end{equation}
Here $C_1,C_2, C_3>0$ are as in the proof of Theorem \ref{Thm:boundProp}, $a$ and $c_W$ are as in (\ref{W:dec}),
and $D_3>0$ is as in Lemma~\ref{lemma:BoundFourier}. 
\end{lemma}

\begin{proof}
We first note that
\begin{equation}
K_0(f,x) \leq  \| W \|_1 ( \varphi_0^{\sigma} * |f|)(x) + 2 \big( (|W|*\varphi_0^{\sigma}) *|f| \big)(x)
\end{equation}
which may be further estimated as in (\ref{kern_bd}). 

Now, for $t >0$, we recall the bound in (\ref{ft_to_Gauss_est}):
\begin{equation} \label{f_to_gauss}
\big|\< e^{-it H_1}f, \varphi^\sigma_x\>\big|
\leq C_1  e^{C_{2} t|\ln t|} \int_{\R^d} \d y\,e^{-\frac { C_{3} }{t^2+ 1}  |x-y|} |f(y)|
=C_1  e^{C_{2} t|\ln t|} (G_{a_t}*|f|)(x)
\end{equation}
where we use the notation in (\ref{exp_dec}) and set 
\begin{equation}
a_t = \frac{C_3}{t^2+1} 
\end{equation}  
In this case, the bound
\begin{align}
K_t(f,x)&= \|W\|_1 |\<e^{-it H_1}f, \varphi^\sigma_x\>| + 2 \big( |W|* |\<e^{-it H_1}f,\varphi_{(\cdot)}^\sigma\>|\big)(x)\notag\\
&\leq  C_1 e^{C_{2} t|\ln t| } \|W\|_1 \big( G_{a_t}*|f|\big)(x) +  2 C_1 e^{C_{2} t|\ln t| } \big( |W|* \big(G_{a_t}*|f|\big)\big)(x)  \label{eq2.3}
\end{align}
is clear. Moreover, the exponential decay of $W$, see (\ref{W:dec}) in Assumption~\ref{Assump:W}, implies
\begin{align}
\nonumber \big( |W|* \big(G_{a_t}*|f|\big)\big)(x) 
&\leq c_{W} \int_{\R^d} \d z\, e^{-a|x-z|} \big(G_{a_t}*|f|\big)(z) \\ 
&\leq c_W\int_{\R^d} \d z\int_{\R^d} \d y\, e^{-4 b_t|x-z|}e^{-4b_t|z-y|}|f(y)| \label{eq2.2}
\end{align}
where we have used (\ref{Def:bt}), and in particular, that $4b_t \leq \min\{a,a_t \}$. By Lemma \ref{lemma:BoundFourier} 
there is $D_3>0$, depending only on $d$, such that
\begin{equation}
\int_{\mathbb{R}^d} dz e^{-4b_t|x-z|}e^{-4b_t|z-y|} \leq \frac{D_3^2}{(4b_t)^{d}} e^{-b_t |x-y|} \, .
\end{equation}
We conclude that 
\begin{equation}
\big( |W|* \big(G_{a_t}*|f|\big)\big)(x)  \leq c_W\frac{D_3^2}{(4b_t)^{d}} (G_{b_t}*|f|)(x)
\end{equation}
and note that (\ref{kern_bd}) follows from the point-wise bound $G_{a_t}(x)\leq G_{b_t}(x)$.
\end{proof}

We now turn to the special case of a Gaussian.

\begin{lemma}\label{sec2:lm:boundK2}
Let $t \geq 0$, $\sigma>0$, and $x,y\in\R^d$. We have
\beq
K_{t}(\varphi_{x}^\sigma, y) \leq 
P_2(t) e^{C_{2} t|\ln t|} G_{c_t}(x-y)
\eeq
where $G_{c_t}$ is as in (\ref{exp_dec}) with
\begin{equation}
c_{t}=\frac {aC_3}{16( a(t^2+ 1)+C_3( 1+a \sigma^2))}
\end{equation}
and $P_2: \mathbb{R} \to \mathbb{R}$ is the polynomial of degree $4d$ given by
\begin{equation} \label{Def:P2+c_t}
P_2(t)=\frac{e^{\frac{1}{8 \sigma^2}}}{(2 \pi \sigma^2)^{d/2}} P_1(t)\frac{D_3^2}{(4c_t)^d} \, .
\end{equation} 
Here we use freely the notation established in Lemma~\ref{lm2.1}. 
\end{lemma}

\begin{proof}
Applying Lemma \ref{lm2.1} and the simple bound $u^2 \geq u - 1/4$, for all $u \in \mathbb{R}$, we find that 
\begin{align}
K_{t}(\varphi^{x}_\sigma, y)
&\leq P_1(t) e^{C_{2} t|\ln t|} \big(G_{b_t}*\varphi_{x}^\sigma\big)(y)\notag\\
&= \frac {P_1(t)e^{\frac1{8 \sigma^2}}} {(2\pi \sigma^2)^{d/2}}  e^{C_{2} t|\ln t|}  \int_{\R^d} \d z\, e^{-b_t|y-z|} e^{-\frac{|x-z|}{2\sigma^2}}.\label{sec2:lm:boundK2:eq0}
\end{align}
Recalling also Lemma \ref{lemma:BoundFourier}, the bound
\begin{align}
 \int_{\R^d} \d z\, e^{-b_t |y-z|} e^{-\frac{|x-z|}{2\sigma^2}}
 \leq
 D_3^2\big(\frac{1}{4c_t}\big)^{d}
 e^{- c_t |x-y|},\label{sec2:lm:boundK2:eq1}
\end{align}
follows from the estimate
\beq
4c_t\leq \min\Big\{b_t,\frac 1 {2\sigma^2}\Big\} \, .
\eeq
This proves the assertion. 
\end{proof}

\subsection{Iterating the Bound} \label{sec:it_scheme}

In this section, we will iterate the bound proven in Lemma~\ref{lm:BoundAntiComm}, i.e. (\ref{F_it_bd}), and 
complete the proof of Theorem~\ref{thm:propagationbound}. 

We first note that iteration of (\ref{F_it_bd}) produces, for any $N \in \mathbb{N}$, a bound of the form
\begin{equation} \label{it_to_N}
F_t^{\Lambda}(f,g) \leq \sum_{n=1}^N a_n(t,f,g) + R_N(t,f,g) 
\end{equation}
where for each $1 \leq n \leq N$ the terms
\begin{align}\label{def:anfg}
a_n(t,f,g) = C_\sigma^n \int_0^t \d t_1 \int_{\R^d}\d x_1\,& K_{t-t_1}(f,x_1) \int_0^{t_1} \d t_2\int_{\R^d} \d x_2 \, K_{t_1-t_2} (\varphi_{x_1}^\sigma,x_2)\cdots\notag\\
&\times \int_0^{t_{n-1}}\d t_n\,\int_{\R^d} \d x_n\, K_{t_{n-1}-t_n}(\varphi_{x_{n-1}}^\sigma, x_n) |\<e^{-it_n H_1}\varphi_{x_n}^\sigma,g\>|
\end{align}
and similarly, the remainder is given by
\begin{align}\label{def:RNt}
R_N(t,f,g) = C_\sigma^N \int_0^t \d t_1 \int_{\R^d}\d x_1\,& K_{t-t_1}(f,x_1) \int_0^{t_1} \d t_2\int_{\R^d} \d x_2 \, K_{t_1-t_2} (\varphi_{x_1}^\sigma,x_2)\cdots\notag\\
&\times \int_0^{t_{N-1}}\d t_N\,\int_{\R^d} \d x_N\, K_{t_{N-1}-t_N}(\varphi_{x_{N-1}}^\sigma, x_N) F_{t_N}^{\Lambda}(\varphi_{x_N}^\sigma,g) 
\end{align}

Next, we estimate the terms $a_n(t,f,g)$.
\begin{lemma}\label{lm:boundafg}
Let $t>0$ and $n\in\N$. We find the following bound
\begin{align} \label{a_n_f_g_bd}
a_n(t,f,g)
\leq D(t) \frac{P_3(t)^n}{n!} \int_{\R^d} \d x \int_{\R^d} \d y\, e^{-\frac{c_t|x-y|}{4}}|f(x)||g(y)|
\end{align}
where $P_3 : \mathbb{R} \to \mathbb{R}$ is a polynomial in $t$ of degree $6d+1$ with 
\begin{equation}
P_3(t) = \frac{C_{\sigma} e^{C_2} D_3 t P_2(t)}{c_t^d}  \quad \mbox{and} \quad   D(t) = C_1 e^{C_2} D_3 \frac{P_1(t)}{P_2(t)} e^{C_2 t | \ln t|} \, .
\end{equation}
All quantities appearing above are as in Lemma~\ref{lm2.1} and Lemma~\ref{sec2:lm:boundK2}. 
\end{lemma}
\begin{proof}
Fix $t>0$. For our estimate, it will be convenient to recall some bounds from Section~\ref{sec:est_K_t}.
First, let $f \in L^2( \mathbb{R}^d)$ and $x \in \mathbb{R}^d$. Lemma~\ref{lm2.1} shows that for $0 \leq t_1 \leq  t$,
\begin{eqnarray}
K_{t-t_1}(f, x) & \leq & P_1(t - t_1 ) e^{C_{2} (t-t_1)|\ln (t - t_1)|} (G_{b_{t-t_1}}*|f|)(x) \nonumber \\
& \leq & P_1(t) e^{C_{2}(t-t_1)|\ln (t - t_1)|} (G_{c_t}*|f|)(x)
\end{eqnarray}
where we have used that $P_1$ is increasing in $t$ and that $b_{t-t_1} \geq b_t \geq c_t$. 
Next, an application of Lemma~\ref{sec2:lm:boundK2} shows that for any $0 \leq t_j \leq t_{j-1} \leq t$ and all $x, y \in \mathbb{R}^d$,
\begin{eqnarray}
K_{t_{j-1}-t_j}(\varphi_{x}^\sigma, y) & \leq & P_2(t_{j-1} - t_j ) e^{C_{2} (t_{j-1}-t_j)|\ln (t_{j-1} - t_j)|} G_{c_{t_{j-1} - t_j}}(x-y) \nonumber \\
& \leq & P_2(t) e^{C_{2}(t_{j-1}-t_j)|\ln (t_{j-1} - t_j)|} G_{c_{t}}(x-y)
\end{eqnarray}
Here we used that the polynomial $P_2(t)$ is increasing in $t$ and the function $c_t$ is decreasing in $t$.  
Similarly, arguing as in (\ref{f_to_gauss}), we see that
\begin{equation}
 |\<e^{-it_n H_1}\varphi_{x_n}^\sigma,g\>| \leq C_1 e^{C_2 t_n | \ln t_n|} (G_{a_{t_n}}*|g|)(x_n) \leq C_1 e^{C_2 t_n | \ln t_n|} (G_{c_t}*|g|)(x_n)
\end{equation}
for any $0 \leq t_n \leq t$ and $x_n \in \mathbb{R}^d$.
As a final observation, note that for parameters 
$0 = t_{n+1} \leq t_n \leq \cdots \leq t_j \leq t_{j-1} \leq \cdots \leq t_1 \leq t_0 =t$, we have the bound
\begin{align} 
\sum_{j=1}^{n+1} (t_{j-1}-t_j) |\ln(t_{j-1}-t_j)| 
&\leq
\sum_{j=1}^{n+1} (t_{j-1}-t_j) \ln(t_{j-1}-t_j)1\{ t_{j-1}-t_j\geq 1\} + n+1\notag\\
&\leq \ln t \, 1\{t\geq 1\} t + n+1\notag\\
&\leq t|\ln t|+ n+1,\label{eq2.112} 
\end{align}
where $1\{\cdot\}$ stands for the indicator function, we used that $x|\ln x|\leq 1$ for all $x\in (0,1]$, and in the argument of the $\ln$,
we use the bound $t_{j-1} - t_j\leq t$, for $j=1,\ldots,n+1$, to produce a telescopic sum.

Putting all this together we find that
\begin{eqnarray}
a_n(t,f,g) & \leq & C_1 C_\sigma^n P_1(t) P_2(t)^{n-1} \int_0^t dt_1 \cdots \int_0^{t_{n-1}} dt_n e^{C_2 \sum_{j=1}^{n+1} (t_{j-1}-t_j) | \ln(t_{j-1} -t_j)|}  \nonumber \\
& \mbox{ } & \quad \times \int_{\mathbb{R}^d} dx_1 \cdots \int_{\mathbb{R}^d} dx_n (G_{c_t}*|f|)(x_1) G_{c_t}(x_1-x_2) \cdots G_{c_t}(x_{n-1}-x_n)  (G_{c_t}*|g|)(x_n) \nonumber \\
\mbox{ } & \leq & C_1 C_\sigma^n P_1(t) P_2(t)^{n-1} e^{C_2 t | \ln t |} e^{(n+1)C_2} \frac{t^n}{n!}  \nonumber \\
& \mbox{ } & \quad \times \int_{\mathbb{R}^d} dx_1 (G_{c_t}*|f|)(x_1) (G_{c_t} *\cdots *G_{c_t}*|g|)(x_1)
\end{eqnarray} 
The latter integral can be further estimated as
\begin{eqnarray}
\int_{\mathbb{R}^d} dx (G_{c_t}*|f|)(x) (G_{c_t} *\cdots *G_{c_t}*|g|)(x) & = & 
\int_{\mathbb{R}^d} dx |f(x)| (G_{c_t} *\cdots *G_{c_t}*|g|)(x) \nonumber \\
& \leq & c_t^d \left( \frac{D_3}{c_t^d} \right)^{n+1} \int_{\mathbb{R}^d} dx  \int_{\mathbb{R}^d} dy e^{- \frac{c_t|x-y|}{4}} |f(x)| |g(y)| .
\end{eqnarray} 
by using the point-wise estimate in Lemma~\ref{lemma:BoundFourier} on the $n+1$-fold convolution
of $G_{c_t}$ with itself. This is the bound claimed in (\ref{a_n_f_g_bd}). 
\end{proof}

We can now complete the proof of Theorem~\ref{thm:propagationbound}.
\begin{proof}[Proof of Theorem~\ref{thm:propagationbound}]
Given (\ref{it_to_N}) and the estimate in Lemma~\ref{lm:boundafg}, i.e. (\ref{a_n_f_g_bd}), it is clear that we need only show that
the remainder term $R_N(t)$, see (\ref{def:RNt}) goes to zero as $N \to \infty$. We will see that this is the case
uniformly for $t$ in compact sets. 

Fix $T>0$ and let $t \in [-T,T]$. We argue as in the proof of Lemma~\ref{lm:boundafg}. 
The only difference between the term $a_N(t,f,g)$ and $R_N(t)$ is the final factor in the integrand: 
more precisely, $|\langle e^{-it_NH_1} \varphi_{x_n}^{\sigma}, g \rangle|$ is replaced with 
$F_{t_N}^{\Lambda}( \varphi_{x_N}^{\sigma}, g)$. In this case, the naive bound
\begin{equation}
F_{t_N}^{\Lambda}( \varphi_{x_N}^{\sigma}, g) \leq 6 \| \varphi_{x_N}^{\sigma} \|_2 \| g \|_2 = 6 \sqrt{C_{\sigma}} \| g \|_2  
\end{equation}
will suffice. In fact, 
\begin{eqnarray}
R_N(t) & \leq & 6 \sqrt{C_{\sigma}} \| g \|_2 C_\sigma^N P_1(t) P_2(t)^{N-1} \int_0^t dt_1 \cdots \int_0^{t_{N-1}} dt_n e^{C_2 \sum_{j=1}^{N} (t_{j-1}-t_j) | \ln(t_{j-1} -t_j)|}  \nonumber \\
& \mbox{ } & \quad \times \int_{\mathbb{R}^d} dx_1 \cdots \int_{\mathbb{R}^d} dx_N (G_{c_t}*|f|)(x_1) G_{c_t}(x_1-x_2) \cdots G_{c_t}(x_{N-1}-x_N)   \nonumber \\
\mbox{ } & \leq & 6 \sqrt{C_{\sigma}} \| g \|_2 C_\sigma^N P_1(t) P_2(t)^{N-1} e^{C_2 t | \ln t |} e^{N C_2} \frac{t^N}{N!}  \nonumber \\
& \mbox{ } & \quad \times \int_{\mathbb{R}^d} dx_N (G_{c_t} *\cdots *G_{c_t}*|f|)(x_N).
\end{eqnarray} 
Now, another application of Lemma~\ref{lemma:BoundFourier} demonstrates that
\begin{eqnarray}
\int_{\mathbb{R}^d} dx (G_{c_t} *\cdots *G_{c_t}*|f|)(x) & \leq & \int_{\mathbb{R}^d} dx \int_{\mathbb{R}^d} dy (G_{c_t}* \cdots G_{c_t})(x-y) |f(y)| \nonumber \\
& \leq &  \left( \frac{D_3}{c_t^d} \right)^N c_t^d \int_{\mathbb{R}^d} dx \int_{\mathbb{R}^d} dy e^{- \frac{c_t|x-y|}{4}} |f(y)| \nonumber \\
& = &  \left( \frac{D_3}{c_t^d} \right)^N c_t^d \| G_{c_t/4} \|_1 \| f \|_1 
\end{eqnarray}
where here we have used that $f \in L^1(\mathbb{R}^d) \cap L^2( \mathbb{R}^d)$. Since all the quantities in these
estimates are explicit, it is clear that
\begin{equation}
\lim_{N \to 0} \sup_{t \in [-T,T]} R_N(t) = 0
\end{equation}
which completes the proof of Theorem~\ref{thm:propagationbound}.
\end{proof}

\section{The infinite-system dynamics. Proof of Theorem \ref{thm:main}}\label{sec:thermlimit:proof}

Our proof of Theorem \ref{thm:main} will make essential use of the following direct consequences of the propagation bounds of 
Theorem~\ref{thm:propagationbound}.

Let $V$ and $W$ satisfy Assumptions \ref{Assump:V} and \ref{Assump:W}, respectively. Then, there exist continuous functions $\wtilde C(\cdot), \wtilde a(\cdot) >0$,
such that with $\varphi_x^\sigma$ the Gaussians introduced in \eq{Def:Gaussian},
and $f\in L^2(\Rl^d)$ of compact support, denoted by $\supp(f)$, we have
\begin{equation}
\Vert \{ \tau^\Lambda_t(a(f)), a^\#(\varphi_z^\sigma)\} \Vert
\leq \Vert f\Vert_1 e^{\wtilde C(|t|)} e^{-\wtilde a(|t|) d(\supp(f),z)}\label{basicbound_2}
\end{equation}
where $a^{\#}( \cdot)$ refers to either an annihilation or creation operator, compare with (\ref{basicbound}).
This estimate follows from Theorem~\ref{thm:propagationbound} and Theorem~\ref{Thm:boundProp} where we again use
arguments as in \eq{sec2:lm:boundK2:eq0} and \eq{sec2:lm:boundK2:eq1}.

Apart from this key estimate, the proof below uses a combination of several ideas introduced
in \cite{bratteli:1997,bravyi:2006a,nachtergaele:2006,nachtergaele:2010}.

\begin{proof}[Proof of Theorem \ref{thm:main}]

We first prove \eq{TLmonomials} for $f\in L^2(\Rl^d)$ of compact support, say $\supp f\subset X\subset \Rl^d$, for some compact $X$.
Let $(\Lambda_n)_{n\geq 1}$ be an increasing sequence of bounded, measurable sets such that $\bigcup_n \Lambda_n = \Rl^d$.
To show that $(\tau^{\Lambda_n}_t(a(f)))_{n\geq 1}$ is Cauchy, uniformly in $t\in[-T,T]$ for $T>0$, note that, for any $\Lambda_m\subseteq \Lambda_n$, 
the operator
$$
W^\sigma_{\Lambda_n}-W^\sigma_{\Lambda_m}=\int_{\Lambda_n\times\Lambda_n\setminus\Lambda_m\times\Lambda_m}  \!\!\!\!\!\!\!\! \d x \d y\, W(x-y)
a^*(\varphi_x^\sigma)a^*(\varphi_y^\sigma)a(\varphi_y^\sigma)a(\varphi_x^\sigma)
$$
is bounded by our assumptions. Hence, the generator of the
strongly continuous dynamics $\tau^{\Lambda_n}_t$ is a bounded perturbation of the generator of $\tau^{\Lambda_m}_t$. In this case, we can
apply the same perturbation formula as in \eq{perturbation_formula} to compare the two dynamics.
The following identity then holds
\bea\label{eq:cauchyid}
&&\tau^{\Lambda_n}_t(a(f)) - \tau^{\Lambda_m}_t(a(f)) \\
&&\nonumber = \frac{i}{2}\int_0^t \d s\int_{\Lambda_n\times\Lambda_n\setminus\Lambda_m\times\Lambda_m}  \!\!\!\!\!\!\!\! \d x \d y\, W(x-y)
\tau^{\Lambda_n}_t\left(\left[
a^*(\varphi_x^\sigma)a^*(\varphi_y^\sigma)a(\varphi_y^\sigma)a(\varphi_x^\sigma), \tau_{t-s}^{\Lambda_m}(a(f))\right]\right).
\eea
Note the identity
\begin{align}
\left[a^*(\varphi_x^\sigma)a^*(\varphi_y^\sigma)a(\varphi_y^\sigma)a(\varphi_x^\sigma), \tau_{t-s}^{\Lambda_m}(a(f))\right]
=&
a^*(\varphi_x^\sigma)a^*(\varphi_y^\sigma)a(\varphi_y^\sigma)\left\{a(\varphi_x^\sigma), \tau_{t-s}^{\Lambda_m}(a(f)) \right\}\notag\\
&-
a^*(\varphi_x^\sigma)a^*(\varphi_y^\sigma)\left\{a(\varphi_y^\sigma), \tau_{t-s}^{\Lambda_m}(a(f)) \right\} a(\varphi_x^\sigma)\notag\\
&+
a^*(\varphi_x^\sigma)\left\{a^*(\varphi_y^\sigma), \tau_{t-s}^{\Lambda_m}(a(f)) \right\} a(\varphi_y^\sigma)a(\varphi_x^\sigma)\notag\\
&- 
\left\{a^*(\varphi_x^\sigma), \tau_{t-s}^{\Lambda_m}(a(f)) \right\}a^*(\varphi_y^\sigma) a(\varphi_y^\sigma)a(\varphi_x^\sigma).
\end{align}
Bounding (\ref{eq:cauchyid}) in norm, applying (\ref{basicbound_2}), and using the symmetry of $W$, we then find for any $T>0$ constants $C$ and $b>0$, such that for every $t\in[-T,T]$ and every $n>m$,
\be
\Vert \tau^{\Lambda_n}_t(a(f)) - \tau^{\Lambda_m}_t(a(f)) \Vert 
\leq  C\Vert \varphi_0^\sigma\Vert^3_2  \Vert f\Vert_1 \int_{\Lambda_n\times\Lambda_n\setminus\Lambda_m\times\Lambda_m}\!\!\!\!\!\!\!\!\!\!\!\!\!\!\!\! \d x\d y\, |W(x-y)|
 e^{-b d(X, x)},
\ee
which converges to $0$ as $n>m\to\infty$ since $|W(x-y)| e^{-b d(X, x)} \in L^1(\R^{2d})$. 
This shows that for compactly supported $f$, the sequence $(\tau^{\Lambda_n}_t(a(f)))_{n\geq 1}$ is Cauchy (in norm) uniformly for $t\in [-T,T]$. 
Thus, the limit exists and gives rise to an isometry from $\cP_{c}$, the set algebraically generated by $\{a(f),a^*(f): f\in L^2(\Rl^d) \text{ of compact support}\}$, into $ \A(L^2(\Rl^d))$. Equation \eqref{eq:cauchyid} can be applied to see this limit is independent of the sequence $(\Lambda_n)$. As $\cP_{c}$ is dense in $\A(L^2(\Rl^d))$, this isometry extends uniquely to a homomorphism, $\tau_t$, of $\A(L^2(\Rl^d))$. It is straightforward to verify
that $\tau_t\circ \tau_s = \tau_{t+s}$ and, in particular, that $\tau_{-t}$ is the inverse of $\tau_t$. Hence, $\{\tau_t\mid t\in\Rl\}$ is a one-parameter group of automorphisms of the CAR algebra.

To prove the strong continuity in $t$, it suffices to note that, for $f\in L^2(\Rl^d)$ and of compact support, 
the continuity of $t\mapsto \tau^{\Lambda_n}_t(a(f)) - \tau^\emptyset_t(a(f))$ carries over to 
the limiting function $t\mapsto \tau_t(a(f)) - \tau^\emptyset_t(a(f))$ due to the uniform convergence on compact intervals. 
Then, since  $\tau^\emptyset_t$ is already known to be strongly
continuous, $ \tau_t(a(f)) $ must be too. Finally, an $\epsilon/3$ argument shows that the strong continuity extends to the full CAR algebra.
\end{proof}

\appendix

\section{Convergence of the {$\sigma\to0$} limit}\label{sec:strong_resolvent_limit}

We prove that, for any fixed finite number of fermions, the UV-regularized dynamics converges to the standard one as $\sigma$ tends to 0.
For this we consider arbitrary, not necessarily bounded, measurable $\Lambda\subset \Rl_d$. When $\Lambda$ is not bounded, the interaction operator $W_\Lambda^\sigma$ is generally unbounded. Therefore, we start by providing a more careful definition of $H_\Lambda^\sigma$. Note that (\ref{def:interaction}) defines a bounded operator $W_{\Lambda;n}^\sigma$ on the $n$-particle subspace $\big( L^2(\R^d)^{\otimes n}\big)^-$ for each $n$. Define the operator $W_{\Lambda;n}^0$ on $\big( L^2(\R^d)^{\otimes n}\big)^-$ to be multiplication by the function $\sum_{1\leq k < l\leq n}  W_\Lambda(x_k-x_l)$. For $\sigma\geq0$, define 
\be
H_{1;n} + W_{\Lambda;n}^\sigma = \sum_{k=1}^n (-\Delta_k + V(x_k)) + W_{\Lambda;n}^\sigma
\ee
acting on $\big( L^2(\R^d)^{\otimes n}\big)^-$. For real-valued $V,W\in L^\infty(\R^d)$ this operator is self-adjoint on the domain $\mathcal{D}(H_{1;n})=(H^2(\R^d)^{\otimes n})^{-}$. With $\sigma\geq0$, let $H_\Lambda^\sigma$ be the operator acting on Fock space by
$$
H_\Lambda^\sigma= \bigoplus_{n=0}^\infty (H_{1;n} + W_{\Lambda;n}^\sigma)
$$
with domain
$$
\mathcal{D}(H_\Lambda^\sigma)=\{ (\psi_n)\in\mathfrak{F}^-: \psi_n\in \mathcal{D}(H_{1;n}) \text{ and }\sum_{n=0}^\infty \|(H_{1;n}+W_{\Lambda;n}^\sigma)\psi_n\|_2^2<\infty\}.
$$
This operator is well-known to be self-adjoint, see e.g. \cite[Exercise 5.43]{weidmann:1980}.
\begin{theorem}\label{thm:srconv}
For real-valued $V,W\in L^\infty(\R^d)$, take $H_1$ as in \eqref{def:H_1}.
Then, for any measurable set $\Lambda\subset \Rl^d$
$$
H_\Lambda^{\sigma}\to H_\Lambda^0
$$
in the strong resolvent sense as $\sigma\downarrow0$.\newline
\end{theorem}

Using a slight modification of \cite[Thm VIII.20(a)]{reed:1980}, the above theorem readily implies:

\begin{corollary}\label{cor:srconv}
For $t\in\R$  we denote by $U_\Lambda^\sigma(t)= e^{-i t H_\Lambda^{\sigma}}$ and $U_\Lambda^0(t)= e^{- i t H_\Lambda^0}$ the unitary groups generated by $H_\Lambda^{\sigma}$ and $H_\Lambda^0$, respectively. Then
$$
\lim_{\sigma\downarrow0}U_\Lambda^\sigma(t)\psi = U_\Lambda^0(t)\psi
$$
for each $\psi\in\mathfrak{F}^-$, uniformly for $t$ in compact subsets of $\Rl$.
\end{corollary}

\begin{remark}
Theorem \ref{thm:srconv} and Corollary \ref{cor:srconv} apply more generally to any self-adjoint operator $H_1$.
\end{remark}

\begin{proof}[Proof of Theorem \ref{thm:srconv}]
We start by reducing the proof Theorem \ref{thm:srconv} to what is essentially the 2-particle situation. Recall that for $\sigma\geq0$,
$$
H_\Lambda^\sigma=\bigoplus_{n=0}^\infty (H_{1;n}+W_{\Lambda;n}^\sigma).
$$
For each $n$, let $\mathcal{D}_n^-=(\mathcal{S}(\Rl^d)^{\otimes n})^- $ be the antisymmetrized $n$-fold tensor product of the Schwarz space $\mathcal{S}(\Rl^d)$, where $\mathcal{S}(\Rl^d)^{\otimes n}=\bigotimes_{j=1}^n \mathcal{S}(\Rl^d)$ is the set of finite linear combinations of functions $\psi$ of the form $\psi(x_1,...,x_n)=\psi_1(x_1)\cdots\psi_n(x_n)$, with each $\psi_j\in\mathcal{S}(\Rl^d)$. Since $\mathcal{S}(\Rl^d)$ is a core for $H_1$, $D_n^-$ is a core for $H_{1;n}$ by the Corollary to \cite[Thm VIII.33]{reed:1980}. It follows that $\mathcal{D}_n^-$ is a core for $H_{1;n}+W_{\Lambda;n}^\sigma$ for every $\sigma\geq0$. If 
\begin{equation} \label{common_core}
\mathcal{D}^-=\{(\psi_n)\in\mathfrak{F}^-: \psi_n\in \mathcal{D}^-_n \text{ and } \exists N \text{ with }\psi_n=0\,\forall n\geq N\},
\end{equation}
then it is not difficult to see that $\mathcal{D}^-$ is a core for $H_\Lambda^\sigma$ for every $\sigma\geq0$ (this is essentially Example 2 in \cite[Sec VIII]{reed:1980}).
 
As is well-known, strong resolvent convergence follows from strong convergence on a common core, see e.g. \cite[Thm VIII.25(a)]{reed:1980}, 
and thus to prove Theorem \ref{thm:srconv}, it suffices to establish that
\begin{equation} \label{s_conv_on_D}
\lim_{\sigma\downarrow0}H_\Lambda^\sigma\psi=H_\Lambda^0\psi
\end{equation}
for every $\psi\in\mathcal{D}^-$. Given the form of $\mathcal{D}^{-}$, see (\ref{common_core}), it is clear that (\ref{s_conv_on_D}) will follow if
\begin{equation}\label{Wslimit}
\lim_{\sigma\downarrow0}W_{\Lambda;n}^\sigma\psi=W_{\Lambda;n}^0\psi \text{\, for every }\psi\in\mathcal{D}_n^-.
\end{equation}

We need only prove (\ref{Wslimit}). To this end, let $b^*(\cdot)$ and $b(\cdot)$ denote the creation and annihilation operators on the full Fock space $\mathfrak{F}=\bigoplus_{n=0}^\infty L^2(\Rl^d)^{\otimes n}$; for any $f\in L^2(\Rl^d)$ and $\psi\in L^2(\Rl^d)^{\otimes n}$,
$$
(b(f)\psi )(x_1,...,x_{n-1})= \sqrt{n}\int_{\Rl^d}\d x f(x)\psi(x,x_1,...,x_{n-1})
$$
and
$$
(b^*(f)\psi)(x_1,...,x_{n+1})= \sqrt{n+1}f(x_1)\psi(x_2,...,x_{n+1}).
$$
 For any $f,g\in L^2(\Rl^d)$, the operator $b^*(f)b^*(g)b(g)b(f)$ is reduced by the $n$-particle subspace, therefore for $\sigma>0$ we may define 
$$
\tilde{W}_{\Lambda;n}^\sigma= \frac{1}{2}\int_{\Rl^d}\int_{\Rl^d}\d x\d y\, W_{\Lambda}(x,y)b^*(\varphi_x^\sigma)b^*(\varphi_y^\sigma)b(\varphi_y^\sigma)b(\varphi_x^\sigma) \Big|_{\bigotimes_{k=1}^nL^2(\Rl^d)},
$$
which is a bounded operator on $L^2(\Rl^d)^{\otimes n}$. Also define $\tilde{W}_{\Lambda;n}^0$ 
by
$$
(\tilde{W}_{\Lambda;n}^0 h)(x_1,...,x_n)=\frac{n(n-1)}{2}W_{\Lambda}(x_1,x_2) h(x_1,...,x_n)
$$
for $h\in L^2(\Rl^d)^{\otimes n}$. 
With these definitions, for every $\sigma\geq0$, 
$$
W_{\Lambda;n}^\sigma=A_n\tilde{W}_{\Lambda;n}^\sigma \Big|_{(\bigotimes_{k=1}^n L^2(\Rl^d))^{-}}
$$ 
where $A_n$ is the antisymmetrization projection $L^2(\Rl^d)^{\otimes n}\to (L^2(\Rl^d)^{\otimes n})^-$. We conclude that: 
if we prove that $\lim_{\sigma\downarrow0}\tilde{W}_{\Lambda;n}^\sigma\psi=\tilde{W}_{\Lambda;n}^0\psi$ for every $\psi$ of the form $\psi(x_1,...,x_n)=\psi_1(x_1)\cdots \psi_n(x_n)$, where $\psi_1,...,\psi_n\in\mathcal{S}(\Rl^d)$, then \eqref{Wslimit} follows. Moreover, since for every $\sigma\geq0$ $\tilde{W}_{\Lambda;n}^\sigma$ acts nontrivially only on the first two particles, it will suffice to prove that $\tilde{W}_{\Lambda;2}^\sigma\to \tilde{W}_{\Lambda;2}^0$ strongly as $\sigma\downarrow0$. 

Let $\sigma >0$ and introduce the function $\Phi^{\sigma}: \mathbb{R}^{2d} \to \mathbb{R}$ by setting
\begin{equation}
\Phi^{\sigma}(x,y) = \varphi_0^{\sigma}(x) \varphi_0^{\sigma}(y) = \frac{1}{(2 \pi \sigma^2)^d}e^{ - \frac{1}{2 \sigma^2}(|x|^2+|y|^2)} \quad \mbox{for any } x,y \in \mathbb{R}^d \, .
\end{equation}
It is clear that $\Phi^{\sigma}$ is $L^1$-normalized, and moreover, a simple calculation shows that for $\psi\in L^2(\Rl^{2d})$,
\begin{align*}
(\tilde{W}_{\Lambda;2}^\sigma &\psi)(x_1,x_2)\\
&=\int_{\Rl^{4d}} \d x \d y\d z_1\d z_2\,W_\Lambda(x,y)\psi(z_1,z_2) \varphi_{z_1}^\sigma(x)\varphi_{x_1}^\sigma(x)\varphi_{z_2}^\sigma(y)\varphi_{x_2}^\sigma(y)\\
&=(\Phi^\sigma*(W_\Lambda (\Phi^\sigma *\psi)))(x_1,x_2).
\end{align*}
We are now ready to conclude the proof of the theorem.

Let $\psi\in\mathcal{S}(\R^d)^{\otimes 2}$. Then
$$
\tilde{W}^\sigma_{\Lambda;2}\psi=\Phi^\sigma*(W_\Lambda(\Phi^\sigma*\psi))=\Phi^{\sigma}*(W_\Lambda\psi)+\Phi^\sigma*(W_\Lambda(\Phi^\sigma*\psi-\psi))
$$
The first term above converges to the desired limit. In fact, convolutions with appropriately scaled $L^1$-functions converge in $L^p$-norm, see e.g. \cite[Thm 8.14 a)]{folland:1999}, and thus since 
$\psi \in L^2( \mathbb{R}^{2d})$,
$$
\Phi^\sigma*(W_\Lambda \psi)\to W_\Lambda\psi \quad \mbox{in } L^2( \mathbb{R}^{2d}) \mbox{ as } \sigma \to 0 \, .
$$
We handle the remainder with Young's inequality, i.e. the bound
$$
\|\Phi^\sigma*(W_\Lambda(\Phi^\sigma*\psi-\psi))\|_2\leq \|\Phi^\sigma\|_1\|W_\Lambda (\Phi^\sigma*\psi-\psi)\|_2\leq \|W\|_\infty \|\Phi^\sigma*\psi-\psi\|_2.
$$
A further application of \cite[Thm 8.14 a)]{folland:1999} shows that
$$
\lim_{\sigma\downarrow0}\|\Phi^\sigma*\psi-\psi\|_2=0
$$
which proves the result.
\end{proof}

\section{Several Fourier transforms}\label{sec:proof:lemma:BoundFourier}

In this section we aim at proving the following Lemma:

 \begin{lemma}\label{lemma:BoundFourier}
 Let  $a>0$, $n\in\N$ with $n\geq2$, and $x\in\R^d$. If $G_a(x)=e^{-a|x|}$, then there exists a constant $D_3>0$ such that
 \begin{align}
(\underbrace{G_a*G_a*\cdots *G_a}_{n-1 \text{ convolutions}})(x) \leq \Big(\frac{D_3} {a^d}\Big)^n a^{d} e^{-\frac {a|x|} 4}.
 \end{align}
 \end{lemma}

To prove this, we first compute several Fourier transforms. We let $\mathcal F$ denote the unitary Fourier transform and $\mathcal F^*$ its inverse.

\begin{lemma}\label{lm:Fourier1}
Let $a>0$ and $G_a(x)= e^{-a |x|}$. Then for $\xi\in\R^d$
\beq
\big(\mathcal F G_a\big)(\xi)=\frac{2^{d/2} \Gamma(\frac{d+1} 2)}{\sqrt \pi} \frac{a }{(a^2+\xi^2)^{\frac{d+1}2}},
\eeq
where $\Gamma$ denotes the Gamma function. 
\end{lemma}

\begin{proof}
Let $\xi\in\R^d$. We compute using the spherical symmetry of $G_a$
\begin{align}
\big(\mathcal F G_a\big)(\xi)
&=\frac{1}{(2\pi)^{d/2}}\int_{\R^d} \d x e^{-a |x|} e^{-i x\xi}\notag\\
&= |\xi|^{\frac{2-d}2} \int_0^\infty  \d r\, r^{\frac{d}2}\,J_{\frac{d-2} 2}(|\xi| r)\,  e^{-a r}
\end{align}
where $J_\nu(y)$ denotes the Bessel function of first kind. Computing this integral using \cite[Sec. 6.621 eq. 1]{gradshteyn:2007}  gives
\begin{align}
\big(\mathcal F G_a\big)(\xi)= \frac{2^{d/2} \Gamma(\frac{d+1} 2)}{\sqrt \pi} \frac{a }{(a^2+|\xi|^2)^{\frac{d+1}2}},
\label{sec2lmeq1}
\end{align}
which is the assertion. 
\end{proof}

\begin{lemma}\label{lm:Fourier2}
Let $a>0$, $\xi\in\R^d$ and 
\beq
H_a(\xi) = \frac a {(a^2+|\xi|^2)^{\frac{d+1} 2}}.
\eeq
Then for $n\in\N$ we obtain 
\beq
\big(\mathcal F^* H^n_a\big)(x)=  \frac{2^{\frac{2-(d+1)n}2} a^{\frac{d-(d-1)n} 2} |x|^{\frac{d(n-1)+n}{2}} K_{\frac{d(n-1)+n}2}(a|x|)}  {\Gamma\big(\frac{(1+d)n}2\big)},
\eeq
where $\Gamma$ denotes the Gamma function and $K_\nu(y)$ the modified Bessel function. 
\end{lemma}

\begin{proof}
We compute, using again the Fourier transform for spherically symmetric functions,
\begin{align}
\big(\mathcal F^* H^n_a\big)(x) 
&=\frac{a^n}{(2\pi)^{d/2}}\int_{\R^d} \d \xi \Big(\frac 1 {a^2+|\xi|^2}\Big)^{\frac{n(d+1)} 2} e^{i x\xi}\notag\\
&=
a^n |x|^{\frac{2-d}2} \int_0^\infty  \d r\, r^{\frac{d}2}\,J_{\frac{d-2} 2}(|x| r)\,  \Big(\frac 1 {a^2+r^2}\Big)^{\frac{n(d+1)} 2},
\label{sec2lmeq2}
\end{align}
where $J_\nu(y)$ is the Bessel function of first kind. Integrating the latter with the help of \cite[Sec. 6.565 eq. 4]{gradshteyn:2007}, gives
\begin{align}
\big(\mathcal F^* H^n_a\big)(x) &=
 \frac{2^{1-\frac{(d+1)n}2} a^{\frac{d-(d-1)n} 2} |x|^{\frac{d(n-1)+n}{2}} K_{\frac{d(n-1)+n}2}(a|x|)}  {\Gamma\big(\frac{(1+d)n}2\big)}.
\end{align}
\end{proof}
Hence, from Lemma \ref{lm:Fourier1} and Lemma \ref{lm:Fourier2} we obtain
\begin{align}\label{eq:Fourier1}
\mathcal F^*\big((\mathcal F G_a)^n\big)(x)=  
\frac{\big(\Gamma(\frac{d+1} 2)\big)^n 2^{\frac{2-n} 2} a^{\frac{d-(d-1)n} 2} |x|^{\frac{d(n-1)+n}{2}} K_{\frac{d(n-1)+n}2}(a|x|)}  {\pi^{n/2}\Gamma\big(\frac{(d+1)n}2\big)}.
\end{align}

\begin{lemma}\label{lemma:ModifiedBessel}
Let $\eta>0$. The modified Bessel function $K_{\eta}$ satisfies for $y>0$ the bound 
\beq
0\leq K_{\eta} (y) \leq \frac {4^\eta }{y^\eta} e^{-\frac y 4} \Gamma(\eta).
\eeq
\end{lemma}

\begin{proof}
We write the modified Bessel function $K_\eta$ as
\beq
K_{\eta} (y) = \int_0^\infty \d t\, e^{-y\cosh(t)} \cosh(\eta t),
\eeq
see \cite[Sec. 8.432 eq. 1]{gradshteyn:2007}. Using $\frac 1 2 e^x\leq \cosh(x)\leq e^x$ valid for all $x\geq0$, we obtain 
\begin{align}
\int_0^\infty \d t\, e^{-y\cosh(t)} \cosh(\eta t)
\leq \int_0^\infty \d t\, e^{-\frac {y e^t} 2} e^{\eta t} 
= \int_{1}^\infty \d u\, e^{-\frac {u y} 2} u^{\eta-1},
\end{align}
where we performed the change of variables $u=e^t$ in the last line. Now, for $u\geq 1$ and $y>0$ we have $e^{-\frac {u y} 2}\leq e^{-\frac y 4}e^{-\frac {u y} 4}$ and therefore
\begin{align}
\int_{1}^\infty \d u\, e^{-\frac {u y} 2} u^{\eta-1}
&\leq 
e^{-\frac y 4} \int_0^\infty\d u\, e^{-\frac {u y} 4} u^{\eta-1}\notag\\
&= \frac{4^\eta}{y^\eta} e^{-\frac y 4} \int_0^\infty\d u\, e^{-u} u^{\eta-1} = \frac{4^\eta}{y^\eta} e^{-\frac y 4}\Gamma(\eta). 
\end{align}
\end{proof}

\begin{proof}[Proof of Lemma \ref{lemma:BoundFourier}]

Starting with \eqref{eq:Fourier1} and using Lemma \ref{lemma:ModifiedBessel} we obtain,
\begin{align}
(G_a*G_a*\cdots *G_a)(x)
&=(2\pi)^{d(n-1)/2}\mathcal F^*\big((\mathcal F G_a)^n\big)(x)\nonumber\\ 
&\leq 
\frac{\big(\Gamma(\frac{d+1} 2)\pi^{\frac{d-1}2}2^{\frac{3d+1}2}\big)^n \pi^{-\frac d 2}2^{\frac 1 2}  a^{-d(n-1)}\Gamma\big(\frac{d(n-1)+n}{2}\big) e^{-\frac {a|x|} 4}}  {\Gamma\big(\frac{(d+1)n}2\big)} \nonumber\\
&\leq D_3^n a^{-d(n-1)} e^{-\frac {a|x|} 4}
\end{align}
for some explicit constant $D_3>0$ depending on $d$, where
 we used that $(2/\pi^d)^{\frac 1 2}\leq 1$ and 
\beq
\frac{\Gamma(\frac{d(n-1)+n}2)}{\Gamma\big(\frac{(d+1)n}2\big)}\leq 1.
\eeq
This gives the assertion. 
\end{proof}

%
%

\end{document}